\documentclass[preprint,authoryear]{elsarticle}
\usepackage{amsmath,amssymb,bbm,mathrsfs}
\usepackage{todonotes}
\usepackage{url}
\usepackage{tikz}
\usepackage{algorithm}
\usepackage{algorithmic}

\newcommand{\av}{\mathbf{a}}
\newcommand{\bv}{\mathbf{b}}
\newcommand{\B}{\mathcal{B}}
\newcommand{\bq}[2]{\beta_{#1}^{(#2)}}
\newcommand{\C}{\mathcal{C}}
\renewcommand{\c}{\mathtt{C}}
\newcommand{\gl}{\mathrm{GL}}
\renewcommand{\H}{\mathcal{H}}
\newcommand{\I}{\mathcal{I}}
\newcommand{\J}{\mathcal{J}}
\newcommand{\K}{\mathcal{K}}
\newcommand{\kk}{\mathbbm{k}}
\newcommand{\lex}{\mathrm{lex}}
\newcommand{\ls}{\mathscr{L}}
\newcommand{\lspan}[1]{\langle{#1}\rangle}
\newcommand{\M}{\mathcal{M}}
\newcommand{\mf}{\mathfrak{m}}
\newcommand{\mult}[1]{\xv_{P}(#1)}
\newcommand{\nmult}[1]{\overline{\xv}_{P}(#1)}
\newcommand{\NN}{\mathbbm{N}}
\renewcommand{\P}{\mathcal{P}}
\newcommand{\PKn}{\mathbbm{P}_{K}^{n}}
\newcommand{\rl}{\mathrm{revlex}}
\newcommand{\TT}{\mathbbm{T}}
\newcommand{\U}{\mathcal{U}}
\newcommand{\V}{\mathcal{V}}
\newcommand{\xv}{\mathbf{x}}
\newcommand{\yv}{\mathbf{y}}

\DeclareMathOperator{\ch}{char}
\DeclareMathOperator{\depth}{depth}
\DeclareMathOperator{\gin}{gin}
\DeclareMathOperator{\lc}{lc}
\DeclareMathOperator{\lt}{lt}
\DeclareMathOperator{\mt}{m}

\DeclareMathOperator{\reg}{reg}
\DeclareMathOperator{\supp}{supp}

\newtheorem{theorem}{Theorem}[section]
\newtheorem{lemma}[theorem]{Lemma}
\newtheorem{proposition}[theorem]{Proposition}
\newtheorem{corollary}[theorem]{Corollary}
\newdefinition{definition}[theorem]{Definition}
\newdefinition{remark}[theorem]{Remark}
\newdefinition{example}[theorem]{Example}
\newproof{proof}{Proof}

\begin{document}

\title{Deterministic Genericity for Polynomial Ideals}
\author[is,th]{A. Hashemi}
\ead{Amir.Hashemi@cc.iut.ac.ir}
\author[ks]{M. Schweinfurter}
\author[ks]{W.M. Seiler}
\ead{seiler@mathematik.uni-kassel.de}
\address[is]{Department of Mathematical Sciences, Isfahan University of
  Technology, Isfahan, 84156-83111, Iran} 
\address[th]{School of Mathematics,
  Institute for Research in Fundamental Sciences (IPM), Tehran, P.O.Box:
  19395-5746, Iran}
\address[ks]{Institut f\"ur Mathematik, Universit\"at Kassel, 34132 Kassel,
  Germany}
\date{\today}

\begin{abstract}
  We consider several notions of genericity appearing in algebraic geometry
  and commutative algebra.  Special emphasis is put on various stability
  notions which are defined in a combinatorial manner and for which a
  number of equivalent algebraic characterisations are provided.  It is
  shown that in characteristic zero the corresponding generic positions can
  be obtained with a simple deterministic algorithm.  In positive
  characteristic, only adapted stable positions are reachable except for
  quasi-stability which is obtainable in any characteristic.
\end{abstract}
\maketitle

\section{Introduction}
\label{sec:introduction}

Genericity appears in many places in algebraic geometry and commutative
algebra, as many results considerably simplify, if one assumes that the
considered ideal is in a sufficiently generic position.  While genericity
is well studied theoretically, its algorithmic side has been treated much
less.  There are two natural questions related to a generic position.  To
apply the corresponding theoretical results in a concrete computation, one
must firstly be able to verify effectively whether a given ideal is in the
considered generic position.  If this is not the case, one would secondly
like to find a (preferably sparse) linear transformation into generic
position.

From a theoretical point of view, the second goal is easily achieved by
applying a random transformation.  In practise, this will destroy all
sparsity typically present in problems of interest.  Therefore we will
study here deterministic algorithms that give us a reasonable chance to
render a position generic with a fairly sparse transformation.  We make no
claims of getting an optimal solution for this problem.  In one of the very
few articles dealing with such questions, \cite{es:sparse} argue that
different notions of optimality exist.  Furthermore, they showed that
various related problems are NP-complete.

Our main emphasis is on generic positions related to Gr\"obner bases where
the leading ideal exhibits certain favourable properties.  The most famous
generic position here is of course the one where the leading ideal is the
generic initial ideal.  However, both above mentioned problems are very
hard for this position and all computer algebra implementations we know of
use random transformations without a check that they have really obtained
the generic initial ideal.

One of the main points of this work lies in studying generic positions --
mainly of a combinatorial nature and related to stability -- that share as
many properties with the generic initial ideal as possible but which are
still effectively verifiable and constructable.  We will present an
algorithm to achieve deterministically any stable position via a sequence
of elementary moves.  The basic idea is due to \cite{wms:geocompl} in the
context of differential equations (the original proof contained a gap and
was later corrected by \cite{wms:comb2}).  There only the case of
quasi-stability was considered and the required moves were selected via a
comparison of Pommaret and Janet multiplicative variables.  Here we present
now a version where the selection criterion is directly based on the
combinatorial characterisation of the various stability notions and which
is therefore no longer restricted to quasi-stability.  While the algorithm
itself is very simple, the termination proof is rather long and technical.

In practise, we face here a conflict: the closer we get to the generic
initial ideal, the harder it becomes to obtain deterministically the
corresponding generic position (meaning the more coordinate transformations
are typically needed).  This observation explains why we consider so many
different kinds of stable positions.  They allow us to make a trade-off: we
go for a generic position that is just enough for the intended application.
One possible application is given by the computation of many fundamental
invariants like the depth, the Castelnuovo-Mumford regularity or the
reduction number which becomes significantly simpler, if the ideal is in a
sufficiently generic position.  For lack of space, we cannot discuss here
details, but refer e.\,g.\ to
\citep{ah:snp,ah:cmr,wms:quasigen,wms:rednum,wms:comb2,wms:noether} and
references therein.  Coordinate transformations to various stable positions
also play a crucial role in recent approaches to determine explicit
equations for Hilbert and Quot schemes -- see \cite{ma:diss} and references
therein.

The Castelnuovo-Mumford regularity nicely exemplifies these considerations.
\cite{bs:mreg} proved that generically the degree of the Gr\"obner basis
for the degree reverse lexicographic order is the regularity.  But no
effective criterion is known to verify whether or not a given ideal is in
generic position and random transformations are the only way to obtain such
a position.  \cite[Ex.~9.9]{wms:comb2} gives a concrete example where the
degree of the Gr\"obner basis is first smaller than the regularity and then
becomes larger after a certain linear coordinate transformation.  Thus the
result of Bayer and Stillman does not even provide a bound.  If the ideal
is in quasi-stable position, then it possesses a finite Pommaret basis the
degree of which is the regularity.  Now the existence of the finite
Pommaret basis provides an effective proof of the genericity of the used
coordinates and the deterministic algorithm provided in this work
effectively constructs such coordinates.

This article is structured as follows.  The next section recalls briefly
some notions and tools; this concerns in particular Pommaret bases and
Gr\"obner systems.  Section~\ref{sec:combgener} discusses the classical
combinatorial concept of stability.  We will introduce a total of nine
different variants of it and provide for all of them equivalent algebraic
characterisations.  Furthermore, we discuss the role of the characteristic
of the base field and componentwise stability.  In Section \ref{sec:other},
we study four other generic positions.  We first show that the classical
Noether position coincides with one of our variants of stability.  To the
best of our knowledge, this represents the first combinatorial
characterisation of Noether position.  Then we very briefly recall some
facts about Borel-fixed ideals and their relation to stability.  We also
provide a deterministic method to compute the generic initial ideal via
Gr\"obner systems.  While we are sure that many people are aware of this
method, we could not find it anywhere in the literature.  Finally, we
introduce the new concept of $\beta$-maximal position and show that it
corresponds to the genericity notion underlying the generic annihilator
numbers.

Section \ref{sec:examples} provides a large number of examples
demonstrating that the various genericity notions are indeed all distinct.
Section \ref{sec:algo} contains the main result of this article from a
computational point of view: a deterministic algorithm to achieve any
variant of stability over fields with characteristic zero.  In positive
characteristic only the positions related to quasi-stability are
effectively reachable (for sufficiently large fields).  For all other
stability notions only adapted ``$p$-versions'' can be used which, however,
lack the algebraic properties of the standard versions.  Finally, we
present the results of some preliminary experiments with the proposed
algorithm.

\section{Preliminaries}
\label{sec:prelim}

We begin by fixing our basic notations and assumptions.  $\P=\kk[\xv]$ with
$\xv=\{x_{1},\dots,x_{n}\}$ will always be the underlying polynomial ring
over a base field $\kk$.  Some of our results will require $\kk$ to be
infinite (or at least sufficiently large), some will depend on whether or
not the characteristic of $\kk$ is positive.  The set of all terms in $\P$
is called $\TT$.  For a non-constant term $\xv^{\mu}\in\TT$, we denote by
$\mt{(\xv^{\mu})}$ the maximal index $k$ such that $\mu_{k}\neq0$.  If
$x^{\mu}=1$, then we set $\mt{(x^{\mu})}=1$.  For simplicity, we consider
exclusively homogeneous ideals $\I\lhd\P$ and thus always assume that all
considered polynomials are homogenous, too.  The homogeneous maximal ideal
in $\P$ is denoted by $\mf=\langle x_{1},\dots,x_{n}\rangle$ and the
saturation of an ideal $\I\lhd\P$ by $\I^{\mathrm{sat}}=\I:\mf^{\infty}$.
Given a finite set $F$ of polynomials, we briefly write $\deg{F}$ for the
maximal degree of an element of $F$.

A term order $\prec$, i.\,e.\ a total order on $\TT$ which is
multiplicative and a well-order, defines for any polynomial $0\neq f\in\P$
its leading term $\lt{f}$ as the maximal term in the support of $f$ with
respect to $\prec$ and we call for any ideal $\I\unlhd\P$ the monomial
ideal $\lt{\I}=\lspan{\lt{f}\mid f\in\I}$ its leading ideal.  If not
explicitly stated otherwise, we will use throughout the degree reverse
lexicographic order (with $x_{1}\succ x_{2}\succ\cdots\succ x_{n}$) for
choosing leading terms, as it has a special relation to the stability
notions studied here.  The use of this order is crucial for the correctness
of our algorithm.

A finite polynomial set $G\subset\I$ is a \emph{Gr\"obner basis} of the
ideal $\I\lhd\P$, if $\lspan{\lt{G}}=\lt{\I}$.  Given a term
$x^{\mu}\in\TT$ with $\mt{(\xv^{\mu})}=k$, we call the variables
$x_{k},\dots,x_{n}$ \emph{(Pommaret) multiplicative} for it and denote them
by $\mult{x^{\mu}}$.  The \emph{non-multiplicative variables} form simply
the complement: $\nmult{x^{\mu}}=\xv\setminus\mult{x^{\mu}}$.  A finite set
$\H\subset\TT$ of terms is a \emph{Pommaret basis}\footnote{See
  \citep{wms:comb1} or \citep{wms:invol} for a general introduction to
  involutive bases, a special kind of Gr\"obner bases with additional
  combinatorial properties to which Pommaret bases belong.  The second
  reference also contains some historical remarks.} of the monomial ideal
$\I=\lspan{\H}$ they generate, if $\I$ can be written as the direct sum
$\I=\bigoplus_{h\in\H}\kk[\mult{h}]\cdot h$.  A finite set $\H\subset\P$ of
polynomials is a Pommaret basis of $\I=\lspan{\H}$, if all its elements
possess pairwise distinct leading terms and $\lt{\H}$ is a Pommaret basis
of $\lt{\I}$.  Obviously, any Pommaret basis is a (generally not reduced)
Gr\"obner basis but not vice versa.

There is a natural action of $\gl(n,\kk)$ on the polynomial ring $\P$ via
linear coordinate transformations:
$x_{i}\mapsto \sum_{j=1}^{n}a_{ij}x_{j}=(A\cdot\xv)_{i}$ for a non-singular
matrix $A=(a_{ij})\in\kk^{n\times n}$.  If we consider the effect of such
coordinate transformations, we always assume that term orders are defined
via exponent vectors and that we use the same term order before and after
the transformation.

For analysing the effect of this $\gl(n,\kk)$-action on a given ideal $\I$,
it is useful to recall the notion of a Gr\"obner system introduced by
\cite{vw:compgb} as part of his theory of comprehensive Gr\"obner bases.
Let $\tilde\P=\P[\av]=\kk[\av][\xv]$ be a \emph{parametric} polynomial ring
with parameters $\av=\{a_1,\ldots,a_m\}$.  Given term orders $\prec_{\av}$
and $\prec_{\xv}$ for terms in the respective variables, we denote by
$\prec_{\xv,\av}$ the corresponding block elimination order with precedence
to the variables $\xv$.

\begin{definition}\label{def:grobsys}
  Let $\tilde\I\unlhd\tilde\P$ be a parametric ideal.  A \emph{Gr\"obner
    system} for $\tilde\I$ for the term order $\prec_{\xv,\av}$ is a finite
  set of triples $\bigl\{(\tilde{G}_i,N_i,W_i)\bigr\}_{i=1}^{\ell}$ with
  finite sets $\tilde{G}_{i}\subset\tilde\P$ and
  $N_{i},W_{i}\subset \kk[\av]$ such that for every index $1\leq i\leq\ell$
  and every specialization homomorphism $\sigma:\kk[\av]\rightarrow\kk$
  with $\sigma(g)=0$ for every $g\in N_i$ and $\sigma(h)\neq 0$ for every
  $h\in W_i$ the set $\sigma(\tilde{G}_i)$ is a Gr\"obner basis of
  $\sigma(\tilde\I)\unlhd\P$ with respect to the order $\prec_{\xv}$ and
  such that for any point $\bv\in\kk^{m}$ an index $1\leq i\leq\ell$ exists
  with $\bv\in\V(N_{i})\setminus\V(\prod_{f\in W_{i}}f)$.
\end{definition}

\citep[Theorem 2.7]{vw:compgb} proved that such a Gr\"obner system exists
for every parametric ideal $\tilde\I\unlhd\P$ and can be effectively
computed.  By now, there exists a number of algorithms and implementations
for this task \cite[e.\,g.\ ][]{ksw:cgs,am:ksw,mw:gbpsp}.  While it is not
part of the definition, every published algorithm for computing Gr\"obner
systems produces systems with an additional property: if two
specialisations $\sigma, \tau$ belong to the same triple
$(\tilde{G}_i,N_i,W_i)$, then they yield the same leading terms
$\lt{\sigma(\tilde{G}_{i})}=\lt{\tau(\tilde{G}_{i})}$.  In the sequel, we
will always assume that we are dealing with Gr\"obner systems possessing
this property.  We also note that it is always possible to prescribe
already at the beginning of the computation some equations or inequations
that the parameters must satisfy.

\begin{remark}\label{rem:finlt}
  Interpreting the entries of a matrix $A=(a_{ij})\in\gl(n,\kk)$ as
  parameters, we can compute a Gr\"obner system
  $\bigl\{(\tilde{G}_i,N_i,W_i)\bigr\}_{i=1}^\ell$ of the parametric ideal
  $\tilde{\I}=A\cdot\I \unlhd \kk[a_{ij}][x_1,\ldots,x_n]$ imposing at the
  start the condition that $\det{(A)}\neq0$.  As such a system is finite by
  definition, we conclude that under linear coordinate transformations any
  ideal $\I\unlhd\P$ possesses only finitely many different leading ideals
  (for a fixed term order).
\end{remark}

\section{Generic Positions Related To Stability} 
\label{sec:combgener}

Stability is a classical combinatorial concept playing an important role in
the theory of monomial ideals and depending on the numbering of the
variables.  There are three basic notions---quasi-stability, stability and
strong stability---forming a natural hierarchy.  For each of them, we
introduce two new weaker versions leading to a total of nine different
stability notions following ideas developed in \citep{ah:cmr,wms:rednum}.
The extension to arbitrary ideals is straightforward via a term order.

\begin{definition}\label{def:stable}
  Let $\J\lhd\P$ be a monomial ideal, $q$ the maximal degree of a
  minimal generator of $\J$ and $0\leq\ell<n$ an integer.
  \begin{description}
  \item[(i)] The ideal $\J$ is \emph{quasi-stable}, if for every term
    $\xv^{\mu}\in\J$ and every index $j<m=\mt{(\xv^{\mu})}$ the term
    $x_{j}^{q}\xv^{\mu}/x_{m}^{\mu_{m}}$ also lies in $\J$.  The ideal is
    \emph{$\ell$-quasi-stable}, if the above condition is satisfied for all
    terms $\xv^{\mu}\in\J$ with $\mt{(\xv^{\mu})}\geq n-\ell$, and
    \emph{weakly $\ell$-quasi-stable}, if the condition is satisfied with
    the additional restriction that $j\leq n-\ell$.
  \item[(ii)] The ideal $\J$ is \emph{stable}, if for every term
    $\xv^{\mu}\in\J$ and every index $j<m=\mt{(\xv^{\mu})}$ the term
    $x_{j}\xv^{\mu}/x_{m}$ also lies in $\J$.  The ideal is
    \emph{$\ell$-stable}, if the above condition is satisfied for all terms
    $\xv^{\mu}\in\J$ with $\mt{(\xv^{\mu})}\geq n-\ell$, and \emph{weakly
      $\ell$-stable}, if the condition is satisfied with the additional
    restriction that $j\leq n-\ell$.
  \item[(iii)] The ideal $\J$ is \emph{strongly stable}, if for every term
    $\xv^{\mu}\in\J$, and every index pair $i>j$ such that
    $x_{i}\mid \xv^{\mu}$ the term $x_{j}\xv^{\mu}/x_{i}$ also lies in
    $\J$.  The ideal is \emph{$\ell$-strongly stable}, if the above
    condition is satisfied for all terms $\xv^{\mu}\in\J$ with
    $\mt{(\xv^{\mu})}\geq n-\ell$ and all indices $i\geq n-\ell$, and
    \emph{weakly $\ell$-strongly stable}, if the condition is satisfied
    with the additional restriction that $j\leq n-\ell$.
  \end{description}
  If $\I\lhd\P$ is an arbitrary polynomial ideal, then we say that $\I$ is
  in a \emph{stable position} for some term order $\prec$ if its leading
  ideal $\lt{\I}$ is stable.  The same terminology is used for any above
  introduced variant of stability.
\end{definition}

It is well-known that the three classical notions of stability are generic
(see e.\,g.\ \citep[Prop.~4.3.8, Cor.~4.3.16]{wms:invol} for the case of
quasi-stability).  Trivial adaptions of the proofs show that all above
considered variants are generic, too.  It should be noted that in the
literature sometimes strongly stable ideals are simply called stable.
Quasi-stable ideals are also called \emph{ideals of nested type} by
\cite{bg:scmr}, \emph{ideals of Borel type} by \cite{hpv:ext} or
\emph{weakly stable ideals} by \cite{cs:reg}.  In the above definition, we
require that all terms in the monomial ideal $\J$ satisfy certain
conditions.  It is straightforward to show that it suffices to verify that
all minimal generators of $\J$ satisfy these conditions.  Furthermore, we
note the obvious hierarchy
\begin{displaymath}
  \text{strongly stable}\quad\Longrightarrow\quad
  \text{stable}\quad\Longrightarrow\quad\text{quasi-stable}\,.
\end{displaymath}

%
%
\subsection{Quasi-Stability}
\label{sec:quastab}

\begin{proposition}\label{prop:quasistab_equiv}
  Let $\J\lhd\P$ be a monomial ideal with $\dim{(\P/\J)}=D$. Then the
  following statements are equivalent:
  \begin{description}
  \item[(i)] $\J$ is quasi-stable.
  \item[(ii)] If $\xv^\mu$ is a term in $\J$ with $\mu_j>0$ for some
    $1< j \leq n$, then for each exponent $0<r\leq \mu_j$ and each index
    $1 \leq i < j$ an exponent $s\geq 0$ exists such that
    $x_i^s\xv^\mu/x_j^r$ lies in $\J$.
  \item[(iii)] For all $0\leq j \leq n-1$ we have
    \begin{equation}\label{eq:qs}
      \J:x_{n-j}^\infty=\J:\langle x_1,\ldots,x_{n-j}\rangle^\infty\,.
    \end{equation}
  \item[(iv)] $x_n$ is not a zero divisor on $\P/\J^{\mathrm{sat}}$ and
    $x_{n-j}$ is not a zero divisor on
    $\P/\langle\J,x_n,\ldots,x_{n-j+1}\rangle^{\mathrm{sat}}$ for all
    $0< j < D$.
  \item[(v)] $\J:x_n^\infty=\J^{\mathrm{sat}}$ and for all $0<j<D$ we have 
    \begin{equation}
      \langle\J,x_{n},\dots,x_{n-j+1}\rangle:x_{n-j}^\infty=
      \langle\J,x_{n},\dots,x_{n-j+1}\rangle^{\mathrm{sat}}\,.
    \end{equation}
  \item[(vi)] We have an ascending chain
    $\J:x_{n}^{\infty}\subseteq\J:x_{n-1}^{\infty}\subseteq\cdots
    \subseteq\J:x_{n-D+1}^{\infty}$ and for each $1\leq j\leq n-D$ there
    exists a term $x_{j}^{\ell_{j}}\in\J$.
  \item[(vii)] $\J$ has a finite monomial Pommaret basis.
  \item[(viii)] Let $\B=\{t_{1},\dots,t_{r}\}$ be the minimal basis of $\J$
    sorted degree reverse lexicographically with $t_{1}$ the largest
    generator.  For each index $1\leq i\leq r$ set
    $\J_{i}=\lspan{t_{1},\dots,t_{i-1}}:t_{i}$ and
    $\P_{i}=\kk[x_{1},\dots,x_{\mt{(t_{i})}-1}]$.  Then all the ideals
    $\hat{\J}_{i}=\J_{i}\cap\P_{i}\unlhd\P_{i}$ are zero-dimensional.
  \item[(ix)] Every associated prime ideal of $\P/\J$ is of the form
    $\langle x_{1},x_{2},\dots,x_{j}\rangle$ for some index
    $1\leq j\leq n-D$.
  \end{description}
\end{proposition}

\begin{proof}
  Most equivalences are well known and their proofs can e.\,g.\ be found in
  \citep[Prop.~3.2]{bg:scmr}, \citep[Prop.~2.2]{hpv:ext},
  \citep[Prop.~4.4]{wms:comb2}, \citep[Lem.~3.4]{wms:noether}.  Only the
  characterisations \emph{(v)} and \emph{(viii)} are new with \emph{(viii)}
  inspired by ideas of \cite[Sect.~4.1]{gc:diss}.  We therefore prove now
  first that \emph{(iii)} entails \emph{(v)}, then that conversely
  \emph{(v)} entails \emph{(iv)} and finally that \emph{(vii)} and
  \emph{(viii)} are equivalent.

  Assume that the monomial $t\in\P$ satisfies
  $tx_{n-j}^s\in\langle\J,x_{n},\dots,x_{n-j+1}\rangle$ for some integer
  $s>0$ and an index $0\leq j\leq D$.  If $\mt{(t)}>n-j$, then we have
  \begin{displaymath}
    t\in\langle x_{n},\dots,x_{n-j+1}\rangle\subseteq 
    \langle \J,x_{n},\dots,x_{n-j+1} \rangle \subseteq 
    \langle\J,x_{n},\dots,x_{n-j+1}\rangle^{\mathrm{sat}}\,.
  \end{displaymath}
  Otherwise $tx_{n-j}^s\in\J$ and thus
  $t\in\J:x_{n-j}^\infty=\J:\langle x_1,\ldots,x_{n-j}\rangle^\infty$ by
  \emph{(iii)}.  Hence we also find
  $t\in \langle\J,x_{n},\dots,x_{n-j+1}\rangle^{\mathrm{sat}}$.  This
  proves that \emph{(v)} is a consequence of \emph{(iii)}.
	
  Now assume for the second step that $x_{n-j}$ defines a zero divisor in
  the ring $\P/\langle\J,x_n,\ldots,x_{n-j+1}\rangle^{\mathrm{sat}}$ for
  some index $1\leq j \leq D-1$. This means that a polynomial
  $f\notin \langle\J,x_n,\ldots,x_{n-j+1}\rangle^{\mathrm{sat}}$ must exist
  such that
  $x_{n-j}f\in\langle\J,x_n,\ldots,x_{n-j+1}\rangle^{\mathrm{sat}}$ which
  in turn entails the existence of an integer $s$ such that
  $x_{n-j}^{s+1}f\in\langle\J,x_n,\ldots,x_{n-j+1}\rangle$ and thus by
  \emph{(v)} that
  $f\in\langle\J,x_n,\ldots,x_{n-j+1}\rangle:x_{n-j}^\infty=
  \langle\J,x_n,\ldots,x_{n-j+1}\rangle^{\mathrm{sat}}$
  which contradicts the choice of $f$.  Hence \emph{(iv)} follows from
  \emph{(v)}.

  For the proof of the equivalence of \emph{(vii)} and \emph{(viii)}, we
  write $\C_{i}$ for the set of all terms in $\P_{i}$ which are not
  contained in $\J_{i}$ and $k_{i}$ for $\mt{(t_{i})}$.  Thus \emph{(viii)}
  is equivalent to the fact that all these sets are finite.  We will now
  prove that if this is the case, then the Pommaret basis of $\J$ is given
  by the finite set
  \begin{equation}\label{eq:pbqs}
    \H=\B\cup\bigcup_{i=1}^{r}\{st_{i}\mid s\in\C_{i}\}\,.
  \end{equation}
  Obviously, $\H$ generates $\I$ and thus we only have to prove that it is
  involutive for the Pommaret division.  Consider a term
  $r\in\bar{\C}_{i}=\{t_{i}\}\cup\{st_{i}\mid s\in\C_{i}\}$; obviously,
  $\mt{(r)}=k_{i}$. We choose an index $1\leq j<k_{i}$ which is thus
  non-multiplicative for $r$.  If $x_{j}r\in\bar{\C}_{i}$, then there is
  nothing to prove.  Otherwise write $r=st_{i}$ with $s=1$ or $s\in\C_{i}$.
  Then $x_{j}r\notin\bar{\C}_{i}$ is equivalent to $x_{j}s\notin\C_{i}$
  which in turn implies that $x_{j}r\in\lspan{t_{1},\dots,t_{i-1}}$.  Let
  $1\leq\ell<i$ be the smallest index such that $t_{\ell}\mid x_{j}st_{i}$
  and write $x_{j}r=r_{m}r_{nm}t_{\ell}$ with terms
  $r_{m}\in\kk[x_{k_{\ell}},\dots,x_{n}]$ and
  $r_{nm}\in\kk[x_{1},\dots,x_{k_{\ell}-1}]$.  Because of the minimality of
  the index $\ell$, we must have that $r_{nm}\in\C_{\ell}$.  Hence
  $r_{nm}t_{\ell}$ is an element of $\H$ and an involutive divisor of
  $x_{j}r$ so that we are done. \qed
\end{proof}

The following two results generalise some of the characterisations in
Proposition \ref{prop:quasistab_equiv} to the above introduced weaker forms
of quasi-stability and thus provide also for these algebraic
interpretations.

\begin{proposition}\label{prop:lqs:equiv}
  Let $\J \lhd \P$ be a monomial ideal and $\ell$ an integer. Then the
  following statements are equivalent.
  \begin{description}
  \item[(i)] $\J$ is $\ell$-quasi-stable.
  \item[(ii)] If $\xv^\mu\in\J$ satisfies $\mt{(\xv^\mu)}\geq n-\ell$ and
    $\mu_j>0$ for some $n-\ell \leq j \leq n$, then for each
    $0<r\leq \mu_j$ and $1 \leq i < j$ an integer $s\geq 0$ exists such
    that $x_i^s\xv^\mu/x_j^r\in\J$.
  \item[(iii)] For all $0\leq j \leq \ell$ we have
    \begin{equation}\label{eq:lqs}
      \J:x_{n-j}^\infty=\J:\langle x_1,\ldots,x_{n-j}\rangle^\infty\,.
    \end{equation}  			
  \end{description}
\end{proposition}

\begin{proof}
  Assume first that $\J$ is $\ell$-quasi-stable and denote by $B$ its
  minimal basis.  Let $\xv^\mu\in\J$ be a term with $\mu_j>0$ for some
  $n-\ell\leq j\leq n$ and $r$ an integer with $0<r\leq\mu_j$.  Hence
  $k=\mt{(\xv^\mu)}\geq j$.  We want to prove \emph{(ii)} by showing that
  $x_i^{\deg{B}}\xv^\mu/x_j^r$ lies in $\J$ for all integers $i<j$.  By the
  definition of $\ell$-quasi-stability,
  $x_i^{\deg{B}}\xv^\mu/x_k^{\mu_k}\in\J$ for $i< k$.  Therefore there
  exists a term $\xv^{\nu^{(1)}}\in B$ with
  \begin{equation}\label{eq:lqs_div}
    \xv^{\nu^{(1)}}\mid x_i^{\deg{B}}\frac{\xv^\mu}{x_k^{\mu_k}}
  \end{equation}
  and
  $k_1=\mt{(\xv^{\nu^{(1)}})}\leq
  \mt{(x_i^{\deg}{B}\xv^\mu/x_k^{\mu_k})}<k$.  Obviously,
  $\nu^{(1)}_\alpha \leq \mu_\alpha$ for all $i\neq \alpha < k$ and
  $\nu^{(1)}_i \leq \mu_i + \deg{B}$.  Again it follows from the assumed
  $\ell$-quasi-stability that
  $x_i^{\deg{B}}\xv^{\nu^{(1)}}/x_{k_1}^{\nu^{(1)}_{k_1}}\in\J$ and thus
  there exists a term $\xv^{\nu^{(2)}}\in B$ with
  $\xv^{\nu^{(2)}}\mid
  x_i^{\deg}{B}\xv^{\nu^{(1)}}/x_{k_1}^{\nu^{(1)}_{k_1}}$ and
  $\mt{(\xv^{\nu^{(2)}})}=k_2<k_1$.  Furthermore by \eqref{eq:lqs_div},
  $\xv^{\nu^{(2)}}\mid
  x_i^{2\cdot\deg{B}}\xv^\mu/x_{k_1}^{\nu^{(1)}_{k_1}}x_k^{\mu_k}$
  and---since $\deg{(\xv^{\nu^{(2)}})}\leq \deg{B}$ and
  $\nu^{(1)}_{k_1}\leq \mu_{k_1}$---this entails
  \begin{equation}\label{eq:lqs_div2}
    \xv^{\nu^{(2)}}\mid x_i^{\deg{B}}
    \frac{\xv^\mu}{x_{k_1}^{\mu_{k_1}}x_k^{\mu_k}}\,.
  \end{equation}
  We go on like this until we end up with a term
  $\xv^{\nu^{(\omega)}}\in B$ such that
  $\xv^{\nu^{(\omega)}}\mid x_i^{\deg{B}} \xv^{\nu^{(\omega-1)}}/
  x_{k_{\omega-1}}^{\nu^{(\omega-1)}_{k_{\omega-1}}}$
  and $\mt{(\xv^{\nu^{(\omega)}})}=k_\omega<\cdots < k_1<k$ such that
  $k_{\omega-1}=j$.  Hence the following holds:
  \begin{itemize}
  \item $\nu^{(\omega)}_\alpha=0$ for all $\alpha\geq j>k_\omega$.
  \item
    $\nu^{(\omega)}_\alpha \leq \nu^{(\omega-1)}_\alpha \leq \cdots \leq
    \nu^{(1)}_\alpha \leq \mu_\alpha$ for all $i\neq \alpha <j$.
  \item
    $\nu^{(\omega)}_i \leq \nu^{(\omega-1)}_i \leq \cdots \leq \nu^{(1)}_i
    \leq \mu_i + \deg{B}$.
  \end{itemize}
  Analogously to \eqref{eq:lqs_div} and \eqref{eq:lqs_div2}, we have
  \begin{displaymath}
    \xv^{\nu^{(\omega)}}\mid x_i^{\deg{B}}
    \frac{\xv^\mu}{x_j^{\mu_{j}}x_{k_{\omega-2}}^{\mu_{k_{\omega-2}}}\cdots 
          x_{k_1}^{\mu_{k_1}}x_k^{\mu_k}}
  \end{displaymath}
  which entails that $\xv^{\nu^{(\omega)}}$ divides
  $x_i^{\deg{B}}\frac{\xv^\mu}{x_j^r}$ and we are done.
  
  Now assume \emph{(ii)} and let $t$ be a term such that $x_{n-j}^rt\in\J$
  for some exponent $r$ and index $0\leq j\leq\ell$.  Since
  $\mt{(x_{n-j}^rt)}\geq n-j\geq n-\ell$, it follows from \emph{(ii)} that
  for all $i< n-j \leq \mt{(x_{n-j}^rt)}$ there is an integer $s_i$ such
  that the term $x_{i}^{s_i} x_{n-j}^rt/x_{n-j}^r=x_{i}^{s_i}t$ lies in
  $\J$.  Hence we have the inclusion
  $t\langle x_{1},\dots,x_{n-j}\rangle^{(s_1+\cdots+s_{n-j-1}+r)(n-j)}
  \subseteq\J$
  entailing $t\in\J:\langle x_{1},\dots,x_{n-j}\rangle^\infty$ which shows
  \emph{(iii)}.

  Finally assume that the equality \eqref{eq:lqs} holds and consider a term
  $\xv^\mu\in\J$ such that $\mt{(\xv^\mu)}=n-j$ with $j\leq\ell$.  Because
  of \eqref{eq:lqs}, we have
  $\xv^\mu/x_{n-j}^{\mu_{n-j}}\in\J:x_{n-j}^\infty= \J:\langle
  x_1,\ldots,x_{n-j}\rangle^\infty$.
  Hence there is an integer $s$ such that
  $(\xv^\mu/x_{n-j}^{\mu_{n-j}})\langle
  x_1,\ldots,x_{n-j}\rangle^s\subseteq \J$.
  But this inclusion means that for every index $1\leq i< n-j$ a minimal
  generator $t_{i}$ of $\J$ exists which divides
  $x_{i}^s\xv^\mu/x_{n-j}^{\mu_{n-j}}\in\J$.  Because of
  $\deg_{x_{i}}{t_{i}}\leq\deg{B}$, it is clear that we may choose
  $s\leq\deg{B}$ which finally shows that $\J$ is $\ell$-quasi-stable.\qed
\end{proof}

\begin{corollary}\label{cor:dqs_qs}
  Let $\J\unlhd\P$ be a monomial and $\ell$-quasi-stable ideal.  If
  $\ell\geq D-1$ where $D=\dim{(\P/\J)}$, then $\J$ is even quasi-stable.
\end{corollary}

\begin{proof}
  Since the equality
  $\I:x_{n-j}^\infty= \I:\langle x_1,\ldots,x_{n-j}\rangle^\infty$ for all
  $0< j< D$ implies that also
  $\langle\I,x_{n},\dots,x_{n-j+1}\rangle:x_{n-j}^\infty=
  \langle\I,x_{n},\dots,x_{n-j+1}\rangle^{\mathrm{sat}}$
  for all $0 < j < D$, the assertion follows from Propositions
  \ref{prop:quasistab_equiv} and \ref{prop:lqs:equiv}.
\end{proof}

For low-dimensional ideals, this observation significantly reduces the
computational costs of checking quasi-stability.  However, its
straightforward application requires the knowledge of the dimension of the
ideal.  The following simple Algorithm \ref{alg:dqstest} verifies whether a
given monomial ideal is $D$-quasi-stable without a priori knowledge of $D$.
It is an adaption of a similar algorithm for checking $D$-stability
presented in \citep[Alg.~1]{wms:rednum}.  We will prove its correctness
later (Proposition \ref{prop:algdqs}).

\begin{algorithm}
\caption{\textsc{DQS-Test}: Test for
  $D$-quasi-stability}\label{alg:dqstest} 
\begin{algorithmic}[1]
  \REQUIRE minimal basis $G=\{t_{1},\dots,t_{r}\}$ of monomial ideal
           $\J\lhd\P$ 
  \ENSURE The answer to: is $\J$ $D$-quasi-stable?
  \STATE $\ell\leftarrow$ smallest $j$ such that $x_\alpha^{\deg G}\in\I$
         for $\alpha=1,\ldots,n-j$\label{line:dqs_purep} 
  \FORALL {$\xv^\mu \in G$ with $k=\mt{(\xv^\mu)}\ge n-\ell$}
    \FOR {$i=1,\ldots ,k-1$} 
      \IF {$x_{i}^{\deg G}\frac{\xv^\mu}{x_k^{\mu_k}}\notin 
           \langle G\rangle$}
        \STATE \textbf{return} false \label{line:dqs_false}
      \ENDIF
    \ENDFOR
  \ENDFOR 
  \STATE \textbf{return} true
\end{algorithmic}
\end{algorithm}

With minor adaptions of the proof given above, one obtains the following
version of Proposition \ref{prop:lqs:equiv} for the weakly
$\ell$-quasi-stable case.  In Section \ref{sec:npos}, we will relate this
notion of stability to Noether position.

\begin{proposition}\label{prop:wlqs:equiv}
  Let $\J \lhd \P$ be a monomial ideal and $\ell$ an integer.  Then the
  following statements are equivalent.
  \begin{enumerate}
  \item $\J$ is weakly $\ell$-quasi-stable
  \item If $\xv^\mu$ in $\J$ with $\mt(\xv^\mu)\geq n-\ell$ and $\mu_j>0$
    for some $n-\ell \leq j \leq n$, then for each $0<r\leq \mu_j$ and
    $1 \leq i \leq n-\ell$ an integer $s\geq 0$ exists such that
    $x_i^s\frac{\xv^\mu}{x_j^r}$ lies in $\J$.
  \item For all $0\leq j \leq \ell$ holds
    \begin{equation}\label{eq:wlqs}
      \J:x_{n-j}^\infty \subseteq 
      \J:\langle x_1,\ldots,x_{n-\ell}\rangle^\infty\,.
    \end{equation}  			
  \end{enumerate}
\end{proposition}

%
%
\subsection{Stability}
\label{sec:stab}

A study of the problem of characterising algebraically the various variants
of stability has already been started by \cite{wms:rednum} because of its
relevance for computing reduction numbers.  For completeness, we first
recall without proof the following result about $\ell$-stability.

\begin{proposition}[{\citep[Prop. 3.5]{wms:rednum}}]\label{prop:lstab}
  The monomial ideal $\J\unlhd\P$ is $\ell$-stable, if and only if it
  satisfies for all $0\leq j\leq\ell$
  \begin{equation}\label{eq:lstab}
    \langle\J,x_{n},\dots,x_{n-j+1}\rangle:x_{n-j}=
    \langle\J,x_{n},\dots,x_{n-j+1}\rangle:\mf\,.
  \end{equation}
\end{proposition}

\cite{wms:rednum} showed furthermore that $D$-stability implies
quasi-stability whereas this is not the case for weak $D$-stability.  The
following novel result provides an analogous characterisation of weak
$\ell$-stability.

\begin{proposition}
  Let $\J\unlhd\P$ be a monomial ideal. If $\J$ is weakly $\ell$-stable,
  then it satisfies for all $0\leq j \leq \ell$ the equality
  \begin{equation}\label{eq:wlstab_colon}
    \langle\J,x_n,\ldots,x_{n-\ell+1}\rangle:x_{n-j}=
    \langle\J,x_n,\ldots,x_{n-\ell+1}\rangle:\mf\,.
  \end{equation}
\end{proposition}

\begin{proof}
  Assume first that $\J$ is weakly $\ell$-stable and let $t$ be a term such
  that $x_{n-j}t\in\langle\J,x_{n},\dots,x_{n-\ell+1}\rangle$ for some
  $j\leq\ell$.  If $\mt{(t)}>n-\ell$, then
  $t\in\langle x_{n},\dots,x_{n-\ell+1}\rangle$ and nothing is to be
  proven.  Otherwise, we have $x_{n-j}t\in\J$ and
  $\mt{(x_{n-j}t)}=n-j\geq n-\ell$.  The weak $\ell$-stability now entails
  that $x_{i}t=x_{i}\frac{x_{n-j}t}{x_{n-j}}\in\J$ for all $i\leq n-\ell$.
  Hence $t\langle x_{1},\dots,x_{n-\ell}\rangle\subseteq\J$ implying
  $t\mf\subseteq\langle\J,x_{n},\dots,x_{n-\ell+1}\rangle$.  Thus we have
  shown the inclusion ``$\subseteq$'' and the converse one is trivial.\qed
\end{proof}

%
%
\subsection{Componentwise Stability}
\label{sec:compstab}

\cite{hh:complin} introduced the notion of a componentwise linear ideal as
a generalisation of the notion of a stable monomial ideal to polynomial
ideals.  Such ideals have many special properties, in particular concerning
their Betti numbers.  If $\I\lhd\P$ is a homogeneous ideal, then we denote
the ideal generated by the homogeneous component $\I_{d}$ by
$\I_{\lspan{d}}=\lspan{\I_{d}}$.  One can now extend every stable position
defined above to a \emph{componentwise stable position} by requiring that
all ideals $\I_{\lspan{d}}$ with $d\geq0$ are simultaneously in the
corresponding stable position.  For monomial ideals, componentwise (strong)
stability is equivalent to ordinary (strong) stability, as the defining
criterion involves only terms of the same degree.  By contrast,
componentwise quasi-stability is a stronger condition than the ordinary
version.  As for polynomial ideals we do not simply consider their leading
ideals but the (polynomial) component ideals $\I_{\lspan{d}}$, for them
componentwise (strongly) stable position is generally also a stronger
condition than its ordinary counterpart.

We will concentrate in the sequel on componentwise quasi-stability, as it
appears to be the most important notion for applications.  For example, if
a componentwise linear ideal is in componentwise quasi-stable position,
then all its Betti numbers can be directly read off from its Pommaret
basis, as this basis induces the minimal resolution of the ideal
\citep[Thm.~9.12]{wms:comb2}.  Another quite remarkable fact about this
position is that it is of all the generic positions considered in this work
the only one which is not automatically implied by the GIN position (see
Definition \ref{def:ginpos} below). Example \ref{ex:23} provides a concrete
counter example.

The following elementary result implies that it is not really necessary to
work componentwise which requires to treat many and rather large bases and
thus is computationally very inefficient.  In the case of (strongly) stable
position only the only-if-part remains true.  However, for the subsequent
results only this direction is needed so that appropriately adapted
versions can be provided.

\begin{lemma}
  Let $\I\lhd\P$ be a homogeneous polynomial ideal.  The ideal
  $\I_{\langle d\rangle}=\lspan{\I_{d}}$ is in quasi-stable position, if
  and only if the ideal $\I_{[d]}=\lspan{\bigcup_{r\leq d}\I_{r}}$ is in
  quasi-stable position.
\end{lemma}

\begin{proof}
  Obviously, $\I_{\langle d\rangle}=\bigl(\I_{[d]}\bigr)_{\geq d}$.  Now
  the claim follows immediately from \citep[Lemma~2.2]{wms:comb2}.\qed
\end{proof}

We now develop a sufficient criterion for an ideal $\I$ to be in
componentwise quasi-stable position which does not require the
consideration of the component ideals $\I_{\langle d\rangle}$ (or
equivalently $\I_{[d]}$).  Such a criterion is important for deciding
componentwise linearity.  Assuming that the ideal $\I$ is already in
quasi-stable position (so that it possesses a Pommaret basis), we can
derive one based on the first syzygies of $\I$.

If the set $\H=\{h_{1},\dots,h_{s}\}$ is a Pommaret basis of $\I$ and
$x_{k}$ is a non-multiplicative variable for the generator
$h_{\alpha}\in\H$, then the product $x_{k}h_{\alpha}$ possesses a unique
involutive standard representation
\begin{equation}\label{eq:isr}
  x_{k}h_{\alpha}=\sum_{\beta=1}^{s}P^{(\alpha;k)}_{\beta}h_{\beta}
\end{equation}
where each non-vanishing coefficient $P^{(\alpha;k)}_{\beta}$ depends only
on variables which are multiplicative for $h_{\beta}$ and satisfies
$\lt{(P^{(\alpha;k)}_{\beta}h_{\beta})}\preceq\lt{(x_{k}h_{\alpha})}$.
\cite{wms:comb2} showed that the corresponding syzygies form a Pommaret
basis of the first syzygy module of $\I$ (for the Schreyer order induced by
$\H$).  Given a degree $d\geq0$ such that $\I_{d}\neq0$, we introduce two
subsets of the Pommaret basis $\H$: the set
$\H_{d}=\{h\in\H\mid\deg{h}\leq d\}$ collects all generators up to degree
$d$ and the set
$\widehat{\H}_{d}=\{\hat{h}\in\H\mid\exists
h\in\H_{d}:\lt{h}\mid\lt{\hat{h}}\}$ contains in addition all higher order
generators which have a leading term divisible by the leading term of an
element of $\H_{d}$.

\begin{proposition}\label{prop:cqs}
  Let $\I\lhd\P$ be a homogeneous ideal in quasi-stable position and
  $d\geq0$ a degree such that $\I_{d}\neq0$.  The ideal $\I_{[d]}$ is in
  quasi-stable position, if in every involutive standard representation
  \eqref{eq:isr} with $h_{\alpha}\in\widehat{\H}_{d}$ all generators
  $h_{\beta}$ with $P^{(\alpha;k)}_{\beta}\neq0$ also lie in
  $\widehat{\H}_{d}$.  In this case, $\widehat{\H}_{d}$ is the Pommaret
  basis of $\I_{[d]}$.
\end{proposition}

\begin{proof}
  We first note that obviously $\I_{[d]}=\lspan{\H_{d}}$.  Then we denote
  by $\widehat{\I}$ the ideal generated by $\widehat{\H}_{d}$.  If the
  condition on the involutive standard representations is satisfied, then
  $\widehat{\H}_{d}$ is the Pommaret basis of $\widehat{\I}$.  As
  obviously, $\I_{[d]}\subseteq\widehat{\I}$, it suffices to show that
  $\I_{[d]}$ cannot be a proper subset of $\widehat{\I}$.  Assume that this
  was the case.  Then there must exist a generator
  $\hat{h}\in\widehat{\H}_{d}$ which is not contained in $\I_{[d]}$.  Let
  $\hat{h}$ be among all such generators the one with the smallest leading
  term with respect to the used term order.  By construction, there exists
  $h\in\H_{d}$ such that $\lt{\hat{h}}=x^{\nu}\lt{h}$ for some term
  $x^{\nu}$.  We consider the polynomial
  $g=\lc{(h)}\hat{h}-\lc{(\hat{h})}x^{\nu}h\in\widehat{\I}$.  It possesses
  an involutive standard representation with respect to the Pommaret basis
  $\widehat{\H}_{d}$ of the form
  $g=\sum_{\hat{f}\in\widehat{\H}_{d}}P_{\hat{f}}\hat{f}$.  Every generator
  $\hat{f}$ with $P_{\hat{f}}\neq0$ must have a leading term smaller than
  $\hat{h}$, as by construction $\lt{g}\prec\lt{\hat{h}}$, and thus must
  lie in $\I_{[d]}$ according to our choice of $\hat{h}$.  But this implies
  that $\hat{h}\in\I_{[d]}$ contradicting our assumption.  Hence
  $\widehat{\I}=\I_{[d]}$ and $\I_{[d]}$ is in quasi-stable position.\qed
\end{proof}

\begin{remark}\label{rem:css}
  The same statement holds for the componentwise (strongly) stable case.
  The only difference is that now we must assume that the ideal $\I$ is
  already in (strongly) stable position; the criterion itself does not
  change.  As in this case the leading terms $\lt{\H}$ form even the
  minimal basis of $\lt{\I}$, we find that $\widehat{\H}_{d}=\H_{d}$ which
  simplifies the application of the criterion.
\end{remark}

As a simple corollary, we find that componentwise quasi-stability is
generic, too, as the intersection of finitely many Zariski open subsets is
still Zariski open.

\begin{corollary}
  For verifying that a homogeneous ideal $\I\lhd\P$ in quasi-stable
  position is even in a componentwise quasi-stable position, it suffices to
  consider only finitely many ideals $\I_{\lspan{d}}$.  If the degree reverse
  lexicographic order is used, then we may restrict to $d\leq\reg{\I}$.
\end{corollary}

\begin{proof}
  It follows from the previous proposition that it suffices to restrict to
  $d\leq q$ where $q$ is the maximal degree of a generator in the Pommaret
  basis of $\I$.  If the degree reverse lexicographic order is used, then
  $q=\reg{\I}$.\qed
\end{proof}

\begin{example}
  The criterion of Proposition \ref{prop:cqs} is not necessary.  Consider
  the ideal
  $\I=\langle x_{1}^{5}, x_{1}x_{2}^{4},
  x_{1}^{3}x_{2}^{3}\rangle\lhd\kk[x_{1},x_{2}]$.  It is quasi-stable and
  its Pommaret basis is given by
  $\H=\bigl\{x_{1}^{5}, x_{1}x_{2}^{4}, x_{1}^{3}x_{2}^{3},
  x_{1}^{2}x_{2}^{4}, x_{1}^{4}x_{2}^{3}\bigr\}$.  The first two generators
  form the set $\H_{5}$, adding the fourth one yields $\widehat{\H}_{5}$.
  Our criterion is \emph{not} satisfied, as we find as involutive standard
  representation $x_{1}h_{4}=x_{2}h_{3}$ and $h_{3}\notin\widehat{\H}_{5}$.
  Nevertheless, one easily verifies that
  $\I_{\langle5\rangle}=\I_{[5]}=\langle x_{1}^{5}, x_{1}x_{2}^{4}\rangle$
  is quasi-stable.
\end{example}

%
%
\subsection{Positive Characteristic}
\label{sec:poschar}

In principle, all above introduced notions of stability are independent of
the characteristic of the base field.  However, when we will discuss in
Section \ref{sec:algo} how to transform a given polynomial ideal into one
of these positions, the characteristic will play a role.  The simplest
restriction will be that for finite base fields we will have to assume that
the field is sufficiently large (the precise meaning of this will become
apparent below).  A more serious restriction will be that in positive
characteristic, we can only guarantee that one can always reach the various
variants of a quasi-stable position.  For stability and strong stability
only adapted ``$p$-versions'' can be reached generally.  The reason is
simply that in positive characteristic many binomial coefficients vanish
and hence many terms cannot be produced via linear transformations.

In order to define these ``$p$-versions'', we need the following notations
-- see e.\,g.\ \citep[\textsection 15.9.3]{de:ca}.  Let $p$ be an arbitrary
prime number.  For two natural numbers $k$, $\ell$, we say
$k\prec_{p}\ell$, if $\binom{\ell}{k}\not\equiv0\mod{p}$.  Given a term
$\xv^{\mu}$ and natural numbers $i>j$ such that $\mu_{i}>0$, we define for
any natural number $s\leq\mu_{i}$ the $s$th \emph{elementary move} as the
term $e_{i,j}^{(s)}(\xv^{\mu})=x_{j}^{s}\xv^{\mu}/x_{i}^{s}$ and this move
is \emph{$p$-admissible}, if and only if $s\prec_{p}\mu_{i}$.

The following definition of ``$p$-versions'' covers only the classical
stability notions.  Of course, it is trivial to extend it to $\ell$- and
weak versions.

\begin{definition}\label{def:pstab}
  Assume that $\ch{\kk}=p$ is positive.  Then a monomial ideal
  $\J\unlhd\P=\kk[\xv]$ is \emph{$p$-stable}, if for every term
  $\xv^{\mu}\in\J$ in it every $p$-admissible move
  $e_{i,j}^{(s)}(\xv^{\mu})$ with $j<i=\mt{(\xv^{\mu})}$ and $s\leq\mu_{i}$
  yields again a term in $\J$.  The ideal $\J$ is \emph{strongly
    $p$-stable}, if in the definition above every index $i$ with
  $\mu_{i}>0$ can be considered.
\end{definition}

As above, it is sufficient to verify the conditions on some finite
generating set of $\J$.  It is easy to see that Definition \ref{def:stable}
is equivalent to requiring that all elementary moves, i.\,e.\ without any
condition on the exponent $s$, stay inside the ideal.  Thus (strong)
$p$-stability is a weaker notion than ``ordinary'' (strong) stability, as
it simply ignores certain elementary moves.

\section{Other Generic Positions}
\label{sec:other}

We now consider three classical generic positions and introduce a new
fourth one.  The material in the Subsections \ref{sec:bfpos} and
\ref{sec:ginpos} is well-known and included only for the sake of
completeness.  Our main point in all cases is the relationship to the
stability positions considered in the previous section.

%
%
\subsection{Noether Position}
\label{sec:npos}

\begin{definition}\label{def:npos}
  The $D$-dimensional ideal $\I\lhd\P$ is in \emph{Noether position}, if
  the variables $x_{1},\dots,x_{D}$ induce a Noether normalisation of $\I$.
\end{definition}

Noether position is a classical concept in commutative algebra.  The
following well-known result provides a simple effective test via Gr\"obner
bases.

\begin{lemma}[{e.\,g.\ \citep[Lem.~4.1]{bg:cmsub}}]
  \label{lem:nother_equiv}
  Let $\I\lhd\P$ be a $D$-dimensional ideal. Then the following statements
  are equivalent:
  \begin{description}
  \item[(i)] $\I$ is in Noether position.
  \item[(ii)] There are integers $s_i$ such that $x_i^{s_i}\in\lt{\I}$ for
    all $1 \leq i\leq n-D$.
  \item[(iii)] $\dim{(\P/\langle \I,x_{n-D+1},\ldots,x_n \rangle)}=0$.
  \item[(iv)] $\dim{(\P/\lt{\langle \I,x_{n-D+1},\ldots,x_n \rangle})}=0$.
  \end{description}		
\end{lemma}

\begin{remark}\label{rem:qsnp}
  \cite{bg:scmr} proved that an ideal $\I$ is quasi-stable, if and only if
  $\I$ and all primary components of $\lt{\I}$ are simultaneously in
  Noether position.  In fact, it is easy to see that quasi-stability
  implies Noether position \citep[Prop.~4.1]{wms:comb2}, which immediately
  implies that the latter is a generic position, too.
\end{remark}

Almost all algorithms proposed so far to get an ideal into Noether position
are probabilistic -- see e.\,g.\ \citep[Algo.~3.4.5]{gp:singular}.  An
exception is the approach of \cite{dr:noether} using Janet bases.
Furthermore, \citep[Sect.~2]{wms:comb2} contains a method to obtain
deterministically quasi-stable position and as mentioned above this entails
Noether position.  However, the result of Bermejo and Gimenez mentioned in
Remark \ref{rem:qsnp} shows that quasi-stability is stronger than Noether
position.  To the best of our knowledge, the following result represents
the first combinatorial characterisation of Noether position.  In
particular, it implies that Noether position can also be achieved with the
deterministic methods which will be presented in Section \ref{sec:algo}.

\begin{theorem}\label{thm:wdqsnp}
  Let $\I\lhd\P$ be a $D$-dimensional ideal.  It is in Noether position, if
  and only if it is in weakly $D$-quasi-stable position.
\end{theorem}

\begin{proof}
  We first note the following simple consequence of the definition of weak
  $D$-quasi-stability for a monomial ideal $\J$ with minimal basis $B$.  If
  the term $\xv^\mu\in\J$ lies in the ideal, then $\J$ also contains any
  term of the form
  $x_{1}^{\mu_{1}+\nu_{1}}\cdots x_{n-\ell}^{\mu_{n-\ell}+\nu_{n-\ell}}$
  with exponents $\nu_{i}$ that are multiples of $\deg{B}$ satisfying
  $\nu_{1}+\cdots+\nu_{n-\ell}= k\deg{B}$ where
  $k=\#\{\mu_{j}\mid j>n-D\wedge\mu_{j}>0\}$.

  Assume now that $\J=\lt{\I}$ is a weakly $D$-quasi-stable ideal.  If
  there exists a term $x^{\mu}\in\J\cap\kk[x_{n-D+1},\dots,x_{n}]$, then we
  can immediately invoke the observation above to conclude that for each
  $1\leq i\leq n-D$ a term $x_{i}^{s_{i}}$ is contained in $\J$, as
  $\mu_{1}=\cdots=\mu_{n-D}=0$.  Thus $\I$ is in Noether position by Lemma
  \ref{lem:nother_equiv}.

  If the intersection $\J\cap\kk[x_{n-D+1},\dots,x_{n}]$ is empty, then the
  $D$-dimensional cone $1\cdot\kk[x_{n-D+1},\dots,x_{n}]$ lies completely
  in the complement of $\J$.  As for a $D$-dimensional ideal it is not
  possible that the complement contains a $(D+1)$-dimensional cone, the
  intersection $\J\cap\kk[x_{i},x_{n-D+1},\dots,x_{n}]$ must be non-empty
  for any index $1\leq i\leq n-D$.  But if $x^{\mu}$ is a term in this
  intersection, then it follows again from the introductory remark that
  also a term $x_{i}^{s_{i}}$ lies in $\J$ and thus that $\I$ is in Noether
  position.  \qed
\end{proof}

With the help of Theorem \ref{thm:wdqsnp}, we can now provide the postponed
proof that Algorithm \ref{alg:dqstest} for testing $D$-quasi-stability is
indeed correct.

\begin{proposition}\label{prop:algdqs}
  Algorithm \ref{alg:dqstest} is correct.
\end{proposition}

\begin{proof}
  We distinguish three cases:
  \begin{enumerate}
  \item $\J$ is $D$-quasi-stable.
  \item $\J$ is not $D$-quasi-stable, but in Noether position.
  \item $\J$ is neither $D$-quasi-stable nor in Noether position.
  \end{enumerate}
  In the first case, Theorem \ref{thm:wdqsnp} entails that $\J$ is in
  Noether position.  Hence the number $\ell$ computed in Line
  \ref{line:dqs_purep} equals $D$ by Lemma \ref{lem:nother_equiv} and we
  will never reach Line \ref{line:dqs_false} by the definition of
  $D$-quasi-stability.  In the second case, we find again $\ell=D$ by the
  same argument.  But as $\J$ is not $D$-quasi-stable there must be an
  obstruction that leads us correctly to Line \ref{line:dqs_false}.  In the
  last case, $\ell$ is greater than $D$ (we know that $\ell\neq D$, since
  $\J$ is not in Noether position; the assumption $\ell < D$ leads to a
  contradiction, since then $D\leq n-(n-\ell) = \ell < D$).  As $\J$ is not
  $D$-quasi-stable, there exists a term $\xv^\mu\in G$ with
  $k=\mt{(\xv^\mu)}\ge n-D>n-\ell$ such that
  $x_i^{\deg G}\frac{\xv^\mu}{x_k^{\mu_k}}\notin \J$ for some $i < k$. Our
  algorithm will detect this obstruction and thus gives again the right
  answer.\qed
\end{proof}

%
%
\subsection{Borel-Fixed Position}
\label{sec:bfpos}

The next generic position which we consider is distinguished from all the
other ones by the fact that it is the only one which depends on the
characteristic of the underlying field $\kk$.  Recall that the subgroup
$\mathfrak{B}\subseteq \gl(n,\kk)$ of all lower triangular invertible
$n\times n$ matrices is called the \emph{Borel group}.  For any integer
$0\leq\ell<n$, we introduce the \emph{$\ell$-Borel group} as the subgroup
$\mathfrak{B}_{\ell}\leq\mathfrak{B}$ consisting of all matrices
$A\in\mathfrak{B}$ such that for $i<n-\ell$ we have $a_{ii}=1$ and
$a_{ij}=0$ for $i\neq j$ (obviously, $\mathfrak{B}_{n-1}=\mathfrak{B}$).

\begin{definition}\label{def:bfpos}
  The monomial ideal $\J\lhd\P$ is \emph{$\ell$-Borel-fixed} for an integer
  $0\leq\ell<n$, if $A\cdot\J=\J$ for all $A\in\mathfrak{B}_{\ell}$.  The
  polynomial ideal $\I\lhd\P$ is in \emph{$\ell$-Borel-fixed position} for
  a term order $\prec$, if $\lt{\I}$ is $\ell$-Borel-fixed.  If $\ell=n-1$,
  then we drop the suffix $\ell$ and simply speak of a \emph{Borel-fixed}
  ideal and position, respectively.
\end{definition}

It is a classical result \citep[e.\,g.\ ][Prop.~4.2.4]{hh:monid} that any
strongly stable ideal is Borel-fixed (which implies that we deal indeed
with a generic position).  In characteristic zero the converse is true,
too.  If the characteristic is a positive prime $p$, then
$\lspan{x_{1}^{p},x_{2}^{p}}\lhd\kk[x_{1},x_{2}]$ is a simple example of a
Borel-fixed ideal which is not strongly stable.  However, it is easy to see
that in any characteristic a Borel-fixed ideal is quasi-stable
\citep[Cor.~2]{bs:reverse}.  \cite[Prop.~9]{wms:rednum} generalised these
assertions: in characteristic zero a monomial ideal $\J$ is
$\ell$-Borel-fixed for some integer $0\leq\ell<n$, if and only if $\J$ is
strongly $\ell$-stable.  In positive characteristic only one direction is
true.

%
%
\subsection{GIN Position}
\label{sec:ginpos}

A classical result proven first by \cite{gall:weier} in characteristic zero
and then later by \cite{bs:reverse} in arbitrary characteristic asserts
that almost all linear changes of coordinates applied to an ideal
$\I\unlhd\P$ lead to the same leading ideal which is then called the
\emph{generic initial ideal} $\gin{\I}$ of $\I$.  Again by
\cite{ag:divstab} in characteristic zero and by \cite{bs:reverse} in
arbitrary characteristic, it was shown that $\gin{\I}$ is always
Borel-fixed.

\begin{theorem}[Galligo]\label{thm:gall}
  For any ideal $\I\unlhd\P$, there exists a nonempty Zariski open subset
  $\U\subseteq\mathrm{GL}(n,\kk)$ such that $\lt{(A\cdot\I)}=\lt{(B\cdot\I)}$
  for all $A,B \in \U$.
\end{theorem}

\begin{definition}\label{def:ginpos}
  The ideal $\I\unlhd\P$ is in \emph{GIN position} (for a term order
  $\prec$), if $\lt{\I}=\gin{\I}$.
\end{definition}

GIN position is the strongest notion of genericity that we consider in this
work.  It implies all other positions with one exception: componentwise
quasi-stability is an independent property (see Example \ref{ex:23} below).
While the GIN position is very popular among theorists, as in it $\I$ and
$\lt{\I}$ share many invariants, it should be noted that neither a simple
effective criterion nor a simple deterministic algorithm is known for it.
As far as we know, all computer algebra systems use a probabilistic
approach to determine $\gin{\I}$ by applying simply one or more random
transformations.  Such a computation may become quite expensive, as it
inevitably leads to dense polynomials for which a Gr\"obner basis must be
computed.  Furthermore, it cannot be easily tested whether or not the
result really is $\gin{\I}$.

If one uses a parametric coordinate transformation instead of a random one,
the computation becomes of course even more expensive, but the result is
guaranteed to be the correct generic initial ideal.  Let $A=(a_{ij})$ be an
$n\times n$ parametric matrix and $\kk(a_{ij})$ the field of fractions of
$\kk[a_{ij}]$.  We consider the ideal
$\hat{\I}=A\cdot\I\unlhd\kk(a_{ij})[x_1,\ldots,x_n]$.  It follows from
Theorem \ref{thm:gall} that $\lt{\hat{\I}}=\gin{\I}$.  Hence a Gr\"obner
basis of $\hat{\I}$ yields immediately $\gin{\I}$.

Alternatively, we consider $\tilde{\I}=\hat{\I}\cap\tilde{\P}$ where
$\tilde{\P}=\kk[a_{ij}][x_1,\ldots,x_n]$ and compute a Gr\"obner system for
$\tilde{\I}\unlhd\tilde{\P}$ imposing at the start the condition that
$\det{(A)}\neq0$.  Again by Theorem \ref{thm:gall}, the generic branch (the
only one for which the set $N_{i}$ is empty) yields as leading ideal
$\gin{\I}$.  Note that for finding the generic branch it is not necessary
to determine the whole Gr\"obner system.  It suffices to follow at each
case distinction the ``not equal zero'' branch.  This is equivalent to a
fraction-free form of computing a Gr\"obner basis of $\hat{\I}$ and in
practise probably more efficient.

Obviously, this approach requires to work with $n^{2}$ parameters.  If one
is interested in the generic initial ideal for the degree reverse lexicographic
term order, then a slight optimisation is possible.  In this case, it
suffices to take for $A$ a lower triangular matrix with all diagonal
entries equal to $1$ and thus one can reduce the number of parameters to
$n(n-1)/2$.  Indeed, any regular matrix $A$ can be written as a
product\footnote{Classically, one uses decompositions $A=LDU$.  But such a
  decomposition for the inverse $A^{-1}$ yields immediately a decomposition
  of our form for $A$.}  $A=UDL$ where $L$ is a lower triangular, $U$ an
upper triangular and $D$ a diagonal matrix and where both $L$ and $U$ have
only ones on the diagonal.  While the transformation induced by $D$ does
trivially not change the leading term of any polynomial for arbitrary term
orders, it follows from the definition of the degree reverse lexicographic order
that here also the transformation induced by $U$ does not affect any
leading term.  Hence we find that $\lt{(A\cdot\I)}=\lt{(L\cdot\I)}$ and it
suffices to work with the matrix $L$.

%
%
\subsection{$\beta$-Maximal Position}
\label{sec:bmpos}

Given a homogeneous ideal $\I\lhd\P$ and a degree $q$ with $\I_{q}\neq0$,
we denote by $\B_q(\I)=(\lt{\I})_q\cap\TT$ the monomial $\kk$-linear basis
of $(\lt{\I})_q$.  We set $\bq{q}{k}(\I)=\#\{t\in\B_q(\I)\mid\mt(t)=k\}$.
Then the \emph{$\beta$-vector} of $\I$ at degree $q$ is defined as
\begin{equation}\label{eq:bv}
  \beta_q(\I)=\bigl(\bq{q}{1}(\I),\ldots,\bq{q}{n}(\I)\bigr)
  \in\NN_{0}^{n}\;.
\end{equation}

\begin{remark}\label{rem:hilbq}
  The $\beta$-vector provides a convenient way to compare the asymptotic
  behaviour of Hilbert polynomials.  We call the set
  \begin{displaymath}
    \lspan{\B_q(\I)}_{P}=\bigoplus_{t\in\B_q(\I)}\kk[\mult{t}]\cdot t
    \subseteq\lspan{\B_q(\I)}
  \end{displaymath}
  the \emph{Pommaret span} of $\B_q(\I)$ and define
  $h_{\I,q}^{P}(s)=\dim_{\kk}{\bigl(\lspan{\B_q(\I)}_{P}\bigr)_{s}}$.  If
  $h_{\I,q}$ denotes the Hilbert function of the monomial ideal
  $\lspan{\B_q(\I)}$, then obviously $h_{\I,q}^{P}(s)\leq h_{\I,q}(s)$ for
  all degrees $s$ and we have $h_{\I,q}^{P}=h_{\I,q}$, if and only if
  $\B_q(\I)$ is the Pommaret basis of the ideal it generates.
  \cite[Prop.~8.2.6]{wms:invol} showed that\footnote{Strictly speaking,
    \cite[Prop.~8.2.6]{wms:invol} covered a slightly different situation
    than we consider here.  In particular, it is there assumed that one
    deals with a Pommaret basis.  However, the adaption to our case here is
    trivial.}
  \begin{equation}\label{eq:hilbq}
    h_{\I,q}^{P}(q+r)=
    \sum_{i=0}^{n-1}\Bigl(
        \sum_{k=i}^{n-1}\frac{s_{k-i}^{(k)}(0)}{k!}\bq{q}{n-k}(\I)\Bigr)r^{i}
  \end{equation}
  where the modified Stirling numbers $s_{i}^{(j)}(\ell)$ are positive
  integers (see \citep[App.~A.4]{wms:invol} for more details).  Thus
  $h_{\I,q}^{P}$ is polynomial beyond degree $q$.  If we write it as
  $\sum_{i}h_{i}r^{i}$, then its coefficient $h_{n-i}$ is a linear
  combination of $\bq{q}{1},\dots,\bq{q}{i}$ with positive coefficients.
  This simple observation entails that if $\I$ and $\J$ are two homogeneous
  ideals such that $\beta_q(\I)\prec_{\mathrm{lex}}\beta_q(\J)$, then
  $h_{\J,q}^{P}(s)<h_{\I,q}^{P}(s)$ for all sufficiently large degrees $s$,
  and motivates the following novel generic position.
\end{remark}

\begin{definition}\label{def:bmpos}
  The homogeneous ideal $\I\lhd\P$ is in \emph{$\beta$-maximal position}
  (for a given term order $\prec$), if we have for all matrices
  $A\in\gl(n,\kk)$ and all degrees $q\geq0$ with $\I_{q}\neq0$ the
  inequality
  \begin{equation}\label{eq:bmax}
    \beta_q(\I)\succeq_{\mathrm{lex}} \beta_q(A\cdot \I)\;.
  \end{equation}
\end{definition}

We will now first show that $\beta$-maximality implies quasi-stability and
then that the generic initial ideal has at all degrees the same
$\beta$-vector as an ideal in $\beta$-maximal position (implying that
$\beta$-maximality is a generic position).  In both cases, the converse
statement is not true.  In particular, $\beta$-maximal position does not
imply GIN position (see for instance Example \ref{ex:21} below).

\begin{proposition}\label{prop:bm_qs}
  The polynomial ideal $\I\lhd\P$ is in quasi-stable position, if and only
  if the following inequality holds for all matrices $A\in\gl(n,\kk)$ and
  all degrees $q\geq\reg{\lt{\I}}$:
  \begin{equation}\label{eq:bqe}
    \beta_q(\I)\succeq_\lex \beta_q(A\cdot\I)
  \end{equation}
\end{proposition}

\begin{proof}
  Let us assume first that $\I$ is in quasi-stable position.  Then for any
  degree $q\geq\reg{\I}=\reg{\lt{\I}}$ the truncation $\I_{\geq q}$ is even
  in stable position \citep[Prop.~9.6]{wms:comb2}.  Hence $\B_q(\I)$ is a
  Pommaret basis of the ideal it generates and we find that
  $h_{\I,q}^{P}=h_{\I,q}$ which immediately implies \eqref{eq:bqe} by
  Remark \ref{rem:hilbq}.

  For the converse, note that \eqref{eq:bqe} implies
  $h_{\I,q}^{P}=h_{\I,q}$, since there always exists a matrix $A$ such that
  $A\cdot\I$ is in quasi-stable position.  These Hilbert functions coincide
  beyond degree $q$, if and only if $\B_q(\I)$ generates a stable ideal and
  thus if $\I_{\geq q}$ is in stable position \citep[Prop.~9.6]{wms:comb2}.
  But then the original ideal $\I$ is in quasi-stable position
  \citep[Lemma~2.2]{wms:comb2}.\qed
\end{proof}

\begin{corollary}\label{cor:bmqs}
  Any polynomial ideal $\I\lhd\P$ in $\beta$-maximal position is in
  quasi-stable position, too.
\end{corollary}

\begin{proposition}\label{prop:bm_gin}
  If the polynomial ideal $\I\lhd\P$ is in GIN position, then $\I$ is also
  in $\beta$-maximal position.  In particular, if $\I$ is in
  $\beta$-maximal position, then $\beta_q(\gin{\I})=\beta_q(\I)$ for all
  degrees $q$ with $\I_{q}\neq0$.
\end{proposition}

\begin{proof}
  We exploit a result derived in the proof of Galligo's Theorem
  \ref{thm:gall} presented by \cite[Thm.~1.27]{mlg:gin}.  For a given
  degree $q$, let the terms $\{t_{1},\dots,t_{s_{q}}\}$ be a $\kk$-basis of
  $\P_{q}$ ordered according to the degree reverse lexicographic order:
  $t_{1}\succ t_{2}\succ\cdots\succ t_{s_{q}}$.  Then there exists a
  Zariski open subset $\U\subseteq\gl(n,\kk)$ such that for all matrices
  $A\in\U$, all degrees $q\geq0$ and all indices $m\leq s_{q}$ the
  dimension of the $\kk$-linear space
  \begin{displaymath}
    \V_{q,m}(A)=\lspan{\lt{(A\cdot\I)_{q}}}_{\kk}\cap 
               \lspan{t_{1},\dots,t_{m}}_{\kk}
  \end{displaymath}
  takes its maximal possible value (thus if $B\notin\U$, then for at least
  some values of $q$ and $m$ we have
  $\dim_{\kk}{\V_{q,m}(B)}<\dim_{\kk}{\V_{q,m}(A)}$ for any $A\in\U$).  Now
  let a $\kk$-basis of $\B_{q}(\gin{\I})$ be given by the terms
  $\{\tilde{t}_{1},\dots,\tilde{t}_{\ell}\}$ and of
  $\B_{q}(\hat{A}\cdot\I)$ for an arbitrarily chosen matrix
  $\hat{A}\in\gl(n,\kk)$ by $\{\hat{t}_{1},\dots,\hat{t}_{\ell}\}$,
  respectively.  In both cases, we assume again that the bases are ordered
  by the degree reverse lexicographic order.  Then the above maximality
  condition implies that $\tilde{t}_{i}\succ\hat{t}_{i}$ for all
  $1\leq i\leq\ell$.  By definition of the degree reverse lexicographic
  order, we thus find that $\mt{(\tilde{t}_{i})}\leq\mt{(\hat{t}_{i})}$ for
  all indices $i$ which is equivalent to
  $\beta_q(\gin{\I})\succeq_\lex \beta_q(\hat{A}\cdot\I)$.\qed
\end{proof}

\begin{remark}\label{rem:bmtest}
  In principle, these results provide us with a deterministic test for
  $\beta$-maximality.  We first check whether or not we are in a
  quasi-stable position.  If this is not the case, the position cannot be
  $\beta$-maximal by Corollary \ref{cor:bmqs}.  Otherwise, we determine
  $\gin{\I}$ deterministically (as discussed in Section \ref{sec:ginpos})
  and then it suffices by Propositions \ref{prop:bm_qs} and
  \ref{prop:bm_gin} to compare the $\beta$-vectors $\beta_{q}(\I)$ and
  $\beta_{q}(\gin{\I})$ for the finitely many degrees $0\leq q<\reg{\I}$.
  Obviously, such a test is rather expensive.  So far, no deterministic
  algorithm for finding a $\beta$-maximal position is known.  One can only
  apply random transformations and then perform the above described check.

  The ideal
  $\I_{1}=\lspan{x_{1}^{2},x_{1}x_{2}+x_{2}^{2},x_{1}x_{3}}\lhd
  \kk[x_{1},x_{2},x_{3}]$ was already considered by
  \cite[Ex.~1.28]{mlg:gin} as an example where the leading ideal is
  strongly stable but nevertheless not the generic initial ideal.  Indeed
  one finds
  $\lt{\I_{1}}=\lspan{x_{1}^{2},x_{1}x_{2},x_{1}x_{3},x_{2}^{3},x_{2}^{2}x_{3}}$,
  whereas
  $\gin{\I_{1}}=\lspan{x_{1}^{2},x_{1}x_{2},x_{2}^{2},x_{1}x_{3}^{2}}$.  It
  is easy to see that these two monomial ideals have different
  $\beta$-vectors and hence $\I_{1}$ is not in $\beta$-maximal position.
  On the other hand, $\I_{2}=\lt{\I_{1}}$ is strongly stable which implies
  $\gin{\I_{2}}=\I_{2}$ and thus $\I_{2}$ is in $\beta$-maximal position.
  This observation shows that two ideals $\I_{1}$ and $\I_{2}$ may have the
  same leading ideal and yet behave differently with respect to
  $\beta$-maximality.  We conclude that there cannot exist a ``simple''
  deterministic algorithm---meaning an algorithm solely based on the
  analysis of leading terms like the one developed in Section
  \ref{sec:algo} for the various notions of stability---that produces a
  $\beta$-maximal position for arbitrary ideals.
\end{remark}

In the context of a Pommaret basis of $\I$, one can roughly interpret
$\beta$-maximality as a condition that generators with more multiplicative
variables should have lower degrees (note, however, that componentwise
quasi-stability admits the same rough interpretation and is nevertheless
independent of $\beta$-maximality---see Examples \ref{ex:4} and \ref{ex:5}
below).  We will now show that this observation can be related to results
by \cite[Sect.~4.3]{hh:monid} on the annihilator numbers of
graded modules.  In particular, we will prove that the genericity concept
underlying their notion of generic annihilator numbers is exactly
$\beta$-maximality.

\begin{definition}
  A linear form $y\in\P_{1}$ is called
  \emph{quasi-regular}\footnote{Following \cite{ah:almreg}, Herzog and Hibi
    use the terminology \emph{almost regular}.  However, the same concept
    was introduced under the name quasi-regularity much earlier in a rather
    unknown letter of Serre appended to \citep{gs:alg}.  Later, the same
    notion was reinvented by \cite{stc:vcm} under the name
    \emph{filter-regular}.} for the graded $\P$-module $\M$, if the graded
  module $0:_{\M}y=\{m\in\M\mid ym=0\}$ is of finite length (i.\,e.\ if
  only finitely many graded components are non-vanishing).  An ordered
  sequence $(y_{1},\dots,y_{k})\subset\P_{1}$ is \emph{quasi-regular} for
  $\M$, if $y_{i}$ is quasi-regular for $\M/\lspan{y_{1},\dots,y_{i-1}}\M$
  for $1\leq i\leq k$.
\end{definition}

In the sequel, we will concentrate for notational simplicity on the case
that $\M=\P/\I$ for a homogeneous ideal $\I$.  However, all results can be
straightforwardly extended to finitely presented modules $\M=\P^{m}/\U$
with a graded submodule $\U$.  The following result by \cite{wms:spencer2}
shows that quasi-regularity is actually just a different way to view
quasi-stability and that quasi-regular sequences of lengths up to
$n=\dim{\P}$ always exist.

\begin{proposition}[{\citep[Thm.~5.2]{wms:spencer2}}]
  The sequence $(x_{n},\dots,x_{2},x_{1})$ is quasi-regular for $\M=\P/\I$,
  if and only if $\I$ is in quasi-stable position.
\end{proposition}

Given a quasi-regular sequence $\yv=(y_{1},\dots,y_{n})$ of length $n$ for
the graded module $\M=\P/\I$, we introduce the graded modules
\begin{align*}
  A_{i-1}(\yv;\M) &= 0:_{\M/\lspan{y_{1},\dots,y_{i-1}}\M}y_{i}\\
  &\cong\bigl(\lspan{\I,y_{1},\dots,y_{i-1}}:y_{i}\bigr)/
        \lspan{\I,y_{1},\dots,y_{i-1}}
\end{align*}
and define the \emph{annihilator numbers} of $\M$ with respect to the
sequence $\yv$ as $\alpha_{ij}(\yv;\M)=\dim_{\kk}{A_{i-1}(\yv;\M)_{j}}$ for
all indices $0\leq i<n$ and $j\geq0$.  The definition of quasi-regularity
implies immediately that only finitely many of these numbers are non-zero.
The following result shows that the annihilator numbers simply encode how
the elements of the Pommaret basis of $\I$ for the degree reverse lexicographic
order distribute over the different degrees and the different numbers of
multiplicative variables.

\begin{theorem}\label{thm:annpb}
  Let the finite set $\H$ be the Pommaret basis of the homogeneous ideal
  $\I\lhd\P$ for the degree reverse lexicographic term order and $\M=\P/\I$.  Then
  for all admissable indices $i$, $j$
  \begin{equation}\label{eq:annij}
    \alpha_{ij}(x_{n},\dots,x_{1};\M)=
    \#\bigl\{h\in\H\mid \mt{(\lt{h})}=n-i\wedge \deg{(h)}=j+1\bigr\}\,.
  \end{equation}
\end{theorem}

\begin{proof}
  Consider the projection
  $\pi:\P=\kk[x_{1},\dots,x_{n}]\rightarrow
  \tilde{\P}=\kk[x_{1},\dots,x_{n-1}]$ defined by $\pi(f)=f|_{x_{n}=0}$.
  It is easy to see that if $\H$ is the Pommaret basis of $\I$ for the
  degree reverse lexicographic order, then $\pi(\H)\setminus\{0\}$ is the Pommaret
  basis of $\pi(\I)$ for the same term order (we find $\pi(h)=0$, if and
  only if $\mt{(\lt{h})}=n$).  Because of the obvious isomorphism
  $\P/\lspan{\I,x_{n}}\cong \tilde{\P}/\pi(\I)$, it thus suffices to
  consider the case $i=0$; the assertion for all other values of the index
  $i$ follows by an easy induction.

  The case $i=0$ requires the analysis of the homogeneous polynomials
  $f\in(\I:x_{n})_{j}\setminus\I_{j}$.  For any such polynomial the product
  $x_{n}f\in\I_{j+1}$ possesses a unique involutive standard representation
  \citep[Thm.~5.4]{wms:comb1}: $x_{n}f=\sum_{h\in\H}P_{h}h$ with
  coefficients $P_{h}\in\kk[\mult{h}]$ satisfying
  $\lt{(P_{h}h)}\preceq\lt{(x_{n}f)}$.  For any generator $h\in\H$ with
  $\mt{(\lt{h})}<n$, we must have $P_{h}\in\lspan{x_{n}}$ whereas
  $\mt{(\lt{h})}=n$ entails $P_{h}\in\kk[x_{n}]$.  The assumption
  $f\notin\I_{j}$ implies that for at least one generator $h\in\H$ with
  $\mt{(\lt{h})}=n$ the coefficient $P_{h}$ is a non-vanishing constant
  (which is only possible if $\deg{h}=j+1$), as otherwise we could divide
  the above involutive standard representation by $x_{n}$ and would obtain
  a standard representation of $f$.  But this observation proves
  immediately our claim for $i=0$.\qed
\end{proof}

Exploiting properties of Pommaret bases, we obtain the following two
results of \cite[Prop.~4.3.4, Thm.~4.3.6]{hh:monid} as trivial corollaries.

\begin{corollary}\label{cor:annpb}
  Let $\I\lhd\P$ be a homogeneous ideal in quasi-stable position and set
  $\M=\P/\I$.
  \begin{description}
  \item[(i)] $\sum_{j\geq0}\alpha_{ij}(x_{n},\dots,x_{1};\M)=0$, if and
    only if $i<\depth{\I}$.
  \item[(ii)] There exists a Zariski open subset $\U\subseteq\gl(n,\kk)$
    such that for all matrices $B\in\U$ the transformed ordered sequence
    $\yv=B\xv$ is again quasi-regular and for all admissible indices $i$,
    $j$ we have the equality
    $\alpha_{ij}(y_{n},\dots,y_{1};\M)=
    \alpha_{ij}(x_{n},\dots,x_{1};\P/\gin{\I})$.
  \end{description}
\end{corollary}

\begin{proof}
  The first assertion follows immediately from Theorem \ref{thm:annpb} and
  the fact that $\depth{\I}=n-t$ with $t$ the maximal value of
  $\mt{(\lt{h})}$ for a generator $h$ in the Pommaret basis of $\I$ for the
  degree reverse lexicographic order \citep[Prop.~3.19]{wms:comb2}.  The second
  assertion follows from Proposition \ref{prop:bm_gin}.\qed
\end{proof}

\cite[Def.~4.3.9]{hh:monid} call both a quasi-regular sequence $\yv$ as in
Corollary \ref{cor:annpb}(ii) and the corresponding annihilator numbers
\emph{generic}.  According to Proposition \ref{prop:bm_gin}, a generic
quasi-regular sequence thus defines a $\beta$-maximal position and vice
versa.  \cite[Sect.~4.3.2]{hh:monid} conclude their discussion of the
annihilator numbers by studying their relationship to the Betti numbers of
$\M$.  All these results follow again immediately from Theorem
\ref{thm:annpb} and the resolution induced by a Pommaret basis
\citep[Thm.~6.1]{wms:comb2}.  In particular, the estimate given by
\cite[Prop.~4.3.12]{hh:monid} is simply a bigraded version of the one
contained in \citep[Thm.~6.1]{wms:comb2}.

\section{Examples}
\label{sec:examples}

The results in the previous sections entail certain relations between the
above introduced generic positions.  They are depicted in the diagram in
Figure~\ref{fig:pos}.  In order to demonstrate that all positions are
indeed different, we compile a series of examples separating them (for a
field of characteristic zero).  The numbers shown in the various fields of
the diagram correspond to the numbering of the examples.  The used
abbreviations should be largly self-explanatory. ``D'' represents the
dimension $D=\dim{\I}$, thus DS denotes $D$-stable ideals and WDS weakly
$D$-stable ideals.  Similarly, ``C'' stands for componentwise and ``Q'' for
quasi.

\begin{figure}[ht]
  \centering
  \begin{tikzpicture}[scale=0.85]
    \draw[line width=1pt, color=red] (1,0) circle (1cm) node{GIN};
    \draw[line width=1pt, color=green] (1,0) circle (2cm) (-0.5,0) node{SS};
    \draw[line width=1pt, color=orange] (1,0) circle (3cm) (-1.5,0) node{S};
    \draw[line width=1pt, color=blue] (1,0) circle (4cm) (-2.5,0) node{DS};
    \draw[line width=1pt, color=gray] (-1.5,2) circle (3.7cm) (-3.5,4) node{CQS};
    \draw[line width=1pt, color=magenta] (-1,-1.5) circle (3.7cm) (-3,-4) node{$\beta$M};
    \draw[line width=1pt, color=olive] (1.5,0) circle (5.3cm) (6,0) node{WDS};
    \draw[line width=1pt, color=black] (-0.8,0) circle (6cm) (-6,0) node{QS};
    \draw[line width=1pt, color=purple] (-7,-7) rectangle (7,7)  (6,6) node[xshift=-1cm]{NP=WDQS};
    \node[fill=black,circle,text=white,scale=0.8,inner sep=0.5pt,minimum size=15pt] at (-6,6) {1} ;
    \node[fill=black,circle,text=white,scale=0.8,inner sep=0.5pt,minimum size=15pt] at (-6,1) {2} ;
    \node[fill=black,circle,text=white,scale=0.8,inner sep=0.5pt,minimum size=15pt] at (6,1) {3} ;
    \node[fill=black,circle,text=white,scale=0.8,inner sep=0.5pt,minimum size=15pt] at (-4,3) {4} ;
    \node[fill=black,circle,text=white,scale=0.8,inner sep=0.5pt,minimum size=15pt] at (-4,-3) {5} ;
    \node[fill=black,circle,text=white,scale=0.8,inner sep=0.5pt,minimum size=15pt] at (-4.1,-0.1) {6} ;
    \node[fill=black,circle,text=white,scale=0.8,inner sep=0.5pt,minimum size=15pt] at (-1,4) {7} ;
    \node[fill=black,circle,text=white,scale=0.8,inner sep=0.5pt,minimum size=15pt] at (-2,-3.3) {8} ;
    \node[fill=black,circle,text=white,scale=0.8,inner sep=0.5pt,minimum size=15pt] at (2,4.5) {9} ;
    \node[fill=black,circle,text=white,scale=0.8,inner sep=0.5pt,minimum size=15pt] at (-3.4,0) {10} ;
    \node[fill=black,circle,text=white,scale=0.8,inner sep=0.5pt,minimum size=15pt] at (4.5,0) {11} ;
    \node[fill=black,circle,text=white,scale=0.8,inner sep=0.5pt,minimum size=15pt] at (3.5,0) {12} ;
    \node[fill=black,circle,text=white,scale=0.8,inner sep=0.5pt,minimum size=15pt] at (2.5,0.5) {13} ;
    \node[fill=black,circle,text=white,scale=0.8,inner sep=0.5pt,minimum size=15pt] at (-0.8,3) {14} ;
    \node[fill=black,circle,text=white,scale=0.8,inner sep=0.5pt,minimum size=15pt] at (-1,-3) {15} ;
    \node[fill=black,circle,text=white,scale=0.8,inner sep=0.5pt,minimum size=15pt] at (-2.5,-0.5) {16} ;
    \node[fill=black,circle,text=white,scale=0.8,inner sep=0.5pt,minimum size=15pt] at (1,2.5) {17} ;
    \node[fill=black,circle,text=white,scale=0.8,inner sep=0.5pt,minimum size=15pt] at (1,-2.5) {18} ;
    \node[fill=black,circle,text=white,scale=0.8,inner sep=0.5pt,minimum size=15pt] at (-1.5,-0.5) {19} ;
    \node[fill=black,circle,text=white,scale=0.8,inner sep=0.5pt,minimum size=15pt] at (1.8,1.5) {20} ;
    \node[fill=black,circle,text=white,scale=0.8,inner sep=0.5pt,minimum size=15pt] at (1.8,-1.3) {21} ;
    \node[fill=black,circle,text=white,scale=0.8,inner sep=0.5pt,minimum size=15pt] at (-0.5,-0.5) {22} ;
    \node[fill=black,circle,text=white,scale=0.8,inner sep=0.5pt,minimum size=15pt] at (1.55,-0.5) {23} ;
    \node[fill=black,circle,text=white,scale=0.8,inner sep=0.5pt,minimum size=15pt] at (0.5,-0.5) {24} ;
  \end{tikzpicture}
  \caption{``Map of Positions''}\label{fig:pos}
\end{figure}
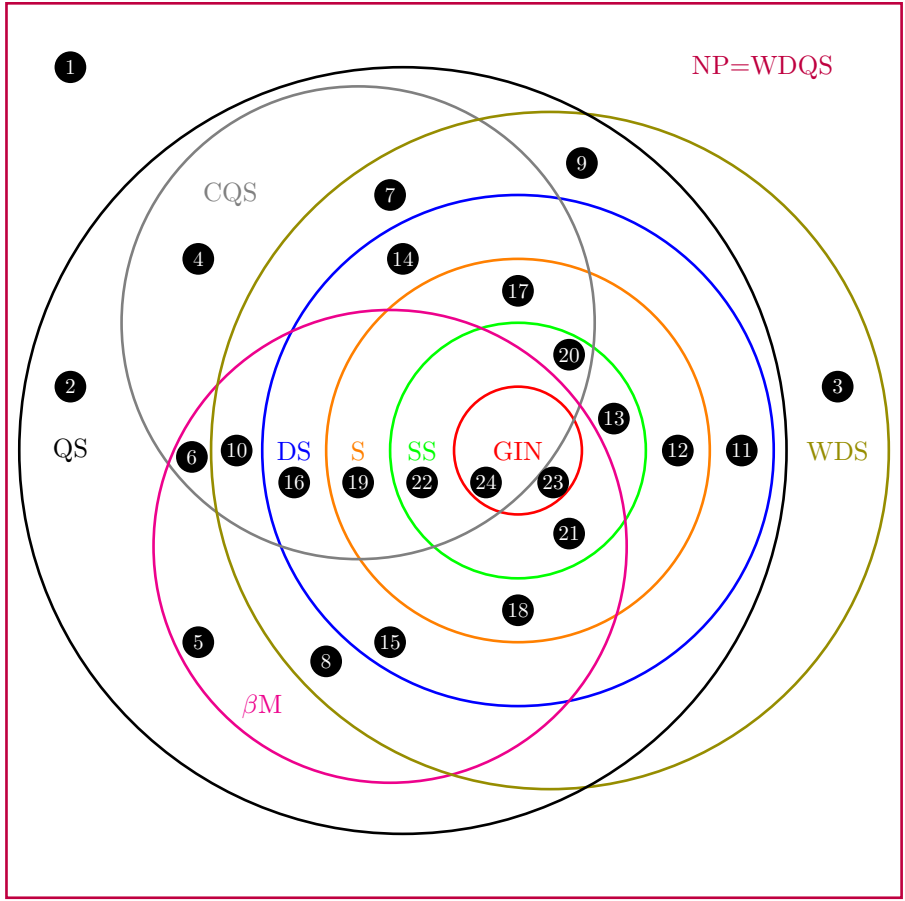

\begin{example}\label{ex:1}
  $\I=\langle x_1^2,x_2^2,x_1x_4\rangle\unlhd\kk[x_1,x_2,x_3,x_4]$ is not
  quasi-stable, because $x_3^2\frac{x_1x_4}{x_4}=x_1x_3^2\notin\I$.  $D=2$
  and $\I$ is not weakly $D$-stable, as
  $x_2\frac{x_1x_4}{x_4}=x_1x_2\notin\I$.  Since
  $\gin{\I}=\langle x_1^2,x_1x_2,x_2^2,x_1x_3^2\rangle$, we see that $\I$
  is not in $\beta$-maximal position as
  $$\beta_2(\I)=(1,1,0,1)\prec_\lex (1,2,0,0)=\beta_2(\gin{\I})\,.$$
\end{example}

\begin{example}\label{ex:2}
  $\I=\langle x_1x_2,x_1^3\rangle\unlhd\kk[x_1,x_2]$ is quasi-stable, but
  not componentwise, as $\I_{\langle 2\rangle}=\langle x_1x_2 \rangle$ is
  not quasi-stable.  $D=1$ and $\I$ is not weakly $D$-stable, as
  $x_1\frac{x_1x_2}{x_2}=x_1^2\notin\I$.  Since
  $\gin{\I}=\langle x_1^2,x_1x_2^2\rangle$, we see that $\I$ is not in
  $\beta$-maximal position, as
  $$\beta_2(\I)=(0,1)\prec_\lex (1,0)=\beta_2(\gin{\I})\;.$$
\end{example}

\begin{example}\label{ex:3}
  $\I=\langle x_1^2,x_1x_3\rangle\unlhd\kk[x_1,x_2,x_3]$ is not
  quasi-stable, as $x_2^2\frac{x_1x_3}{x_3}=x_1x_2^2\notin\I$.  $D=2$ and
  $\I$ is not $D$-stable, as $x_2\frac{x_1x_3}{x_3}=x_1x_2\notin\I$.  Since
  $\gin{\I}=\langle x_1^2,x_1x_2\rangle$, we see that $\I$ is not in
  $\beta$-maximal position, as
  $$\beta_2(\I)=(1,0,1)\prec_\lex (1,1,0)=\beta_2(\gin{\I})\;.$$
\end{example}

\begin{example}\label{ex:4}
  For $\I=\langle x_1^2,x_2^2,x_3^2\rangle\unlhd\kk[x_1,x_2,x_3]$ we have
  $D=0$ and $\I$ is not weakly $D$-stable, as
  $x_1\frac{x_3^2}{x_3}=x_1x_3\notin\I$.
  $\gin{\I}=\langle x_1^2, x_1x_2, x_2^2, x_1x_3^2, x_2x_3^2, x_3^4\rangle$
  implies that $\I$ is not in $\beta$-maximal position, as
  $$\beta_2(\I)=(1,1,1)\prec_\lex (1,2,0)=\beta_2(\gin{\I})$$
\end{example}

\begin{example}\label{ex:5}
  For
  $\I=\langle x_1^3,x_1x_2^2+x_2^2x_3,x_2^4\rangle\unlhd\kk[x_1,x_2,x_3]$,
  we have $\lt{\I}=\langle x_1^3, x_1x_2^2,x_2^4, x_2^2x_3^3\rangle$ and
  $D=1$.  $\lt{\I}$ is not weakly $D$-stable, as
  $x_1\frac{x_1x_2^2}{x_2}=x_1^2x_2\notin\lt{\I}$, and
  $\lt{\I_{\langle 3\rangle}}=\langle x_1^3,x_1x_2^2,x_2^2x_3^3\rangle$ is
  not quasi-stable.  $\I$ is in $\beta$-maximal position, as
  $\gin{\I}=\langle x_1^3, x_1^2x_2, x_1x_2^3, x_2^4, x_1x_2^2x_3^2,
  x_1^2x_3^4\rangle$ and
  \begin{align*}
    \beta_3(\I)&=(1,1,0)=\beta_3(\gin{\I})\;, \\
    \beta_4(\I)&=(1,4,2)=\beta_4(\gin{\I})\;, \\
    \beta_5(\I)&=(1,5,8)=\beta_5(\gin{\I})\;. \\
  \end{align*}
\end{example}

\begin{example}\label{ex:6}
  For $\I=\langle x_1^2,x_2^2\rangle\unlhd\kk[x_1,x_2]$ we have $D=0$ and
  $\I$ is not weakly $D$-stable, as $x_1\frac{x_2^2}{x_2}=x_1x_2\notin\I$.
  $\gin{\I}=\langle x_1^2, x_1x_2, x_2^3\rangle$ implies that $\I$ is
  in $\beta$-maximal position, as
  $$\beta_2(\I)=(1,1)=\beta_2(\gin{\I})\;.$$
\end{example}

\begin{example}\label{ex:7}
  For
  $\I=\langle
  x_1^2,x_1x_2,x_2^2+x_3^2,x_1x_4\rangle\unlhd\kk[x_1,x_2,x_3,x_4]$
  we have $\lt{\I}=\langle x_1^2,x_1x_2, x_2^2, x_1x_4, x_1x_3^2\rangle$
  and $D=2$.  $\lt{\I}$ is not $D$-stable, as
  $x_3\frac{x_1x_4}{x_4}=x_1x_3\notin\I$.
  $\gin{\I}=\langle x_1^2, x_1x_2, x_2^2, x_1x_3, x_1x_4^2\rangle$ entails
  that $\I$ is not in $\beta$-maximal position, as
  $$\beta_2(\I)=(1,2,0,1)\prec_\lex (1,2,1,0)=\beta_2(\gin{\I})\;.$$
\end{example}

\begin{example}\label{ex:8}
  For the ideal
  $\I=\langle
  x_1^3+x_1x_3^2,x_1^2x_2+x_2x_4^2,x_1x_2^2,x_2^3,x_2^2x_3^2,x_2x_3^3
  \rangle$ ${}\unlhd\kk[x_1,x_2,x_3,x_4]$
  we have
  \begin{gather*}
    \lt{\I}=\langle x_1^3, x_1^2x_2, x_1x_2^2, x_2^3,
    x_1x_2x_3^2,x_2^2x_3^2,x_2x_3^3,x_2^2x_4^2, x_1x_2x_3x_4^2,\\
    x_2x_3^2x_4^2,x_1x_2x_4^4, x_2x_3x_4^4, x_2x_4^6\rangle
  \end{gather*}
  and $D=2$.  $\lt{\I}$ is not $D$-stable, as
  $x_3\frac{x_2^2x_4^2}{x_4}=x_2^2x_3x_4\notin\lt{\I}$.  Since
  $\lt{\I_{\langle 3\rangle}}=\langle x_1^3, x_1^2x_2, x_1x_2^2, x_2^3,
  x_1x_2x_3^2,x_2^2x_4^2, x_2x_3^2x_4^2\rangle$
  is not quasi-stable, $\I$ is not in componentwise quasi-stable position.
  \begin{gather*}
    \gin{\I}=\langle x_1^3, x_1^2x_2, x_1x_2^2, x_2^3, x_1^2x_3^2,
    x_1x_2x_3^2, x_1x_3^3, x_1^2x_3x_4, x_1x_2x_3x_4^2,\\ x_1x_3^2x_4^2,
    x_1^2x_4^3, x_1x_2x_4^4, x_1x_3x_4^4, x_1x_4^6\rangle
  \end{gather*}
  entails that $\I$ is in $\beta$-maximal position as
  \begin{align*}
    \beta_3(\I)&=(1,3,0,0)=\beta_3(\gin{\I})\;, \\
    \beta_4(\I)&=(1,4,7,5)=\beta_4(\gin{\I})\;, \\
    \beta_5(\I)&=(1,5,12,20)=\beta_5(\gin{\I})\;, \\
    \beta_6(\I)&=(1,6,18,40)=\beta_6(\gin{\I})\;. \\
  \end{align*}
\end{example}

\begin{example}\label{ex:9}
  For
  $\I=\langle
  x_1^2,x_1x_2,x_2^2,x_1x_4,x_1x_3^2\rangle\unlhd\kk[x_1,x_2,x_3,x_4]$
  we have that
  $\I_{\langle 2\rangle}=\langle x_1^2,x_1x_2,x_2^2,x_1x_4\rangle$ is not
  quasi-stable and thus $\I$ is not in componentwise quasi-stable
  position. $D=2$ and $\lt{\I}$ is not $D$-stable, as
  $x_3\frac{x_1x_4}{x_4}=x_1x_3\notin\I$.
  $\gin{\I}=\langle x_1^2, x_1x_2, x_2^2, x_1x_3, x_1x_4^2\rangle$ entails
  that $\I$ is not in $\beta$-maximal position as
  $$\beta_2(\I)=(1,2,0,1)\prec_\lex (1,2,1,0)=\beta_2(\gin{\I})\;.$$
\end{example}

\begin{example}\label{ex:10}
  Let
  $\I=\langle x_1^3,x_1^2x_2,x_1x_2^2,x_2^3,x_2^2x_3^2,x_2^2x_4^2\rangle
  \unlhd\kk[x_1,x_2,x_3,x_4]$.
  Then $D=2$ and $\I$ is not $D$-stable since
  $x_3\frac{x_2^2x_4^2}{x_4}=x_2^2x_3x_4\notin\I$.  As
  $\gin{\I}=\langle x_1^3, x_1^2x_2, x_1x_2^2, x_2^3, x_1^2x_3^2,
  x_1^2x_3x_4, x_1^2x_4^3\rangle$,
  we see that $\I$ is in $\beta$-maximal position since
  \begin{displaymath}
    \begin{array}{rcl}
      \beta_3(\I)= &(1,3,0,0) & = \beta_3(\gin\I)\,, \\
      \beta_4(\I)= &(1,4,5,5) & = \beta_4(\gin\I)\,.
    \end{array}
  \end{displaymath}
\end{example}

\begin{example}\label{ex:11}
  Let
  $\I=\langle
  x_2^2,x_1x_3,x_2x_3,x_3^2,x_1^3\rangle\unlhd\kk[x_1,x_2,x_3]$.
  Then $\I_{\langle 2\rangle}=\langle x_2^2,x_1x_3,x_2x_3,x_3^2\rangle$ is
  not quasi-stable, hence $\I$ is not in componentwise quasi-stable
  position. Furthermore, $\I$ is not stable, as
  $x_1\frac{x_2^2}{x_2}=x_1x_2\notin\I$.  Since
  $\gin{\I}=\langle x_1^2, x_1x_2, x_1x_3, x_2^2, x_2x_3^2, x_3^4\rangle$,
  we see that $\I$ is not in $\beta$-maximal position, as
  $\beta_2(\I)=(0,1,3)\prec_\lex (1,2,1)=\beta_2(\gin{\I})$.
\end{example}

\begin{example}\label{ex:12}
  Let
  $\I=\langle
  x_1^2,x_1x_2^2+x_2x_3^2,x_2^5,x_2^4x_3,x_2^3x_3^2,x_2^2x_3^3\rangle
  \unlhd\kk[x_1,x_2,x_3]$.
  Then
  $\lt{\I_{\langle 3\rangle}}=\langle x_1^3,x_1^2x_2,x_1x_2^2, x_1^2x_3,
  x_1x_2x_3^2, x_2x_3^4\rangle$
  is not quasi-stable, hence $\I$ is not in componentwise quasi-stable
  position.  Furthermore,
  \begin{displaymath}
    \lt{\I}=\langle x_1^2, x_1x_2^2, x_1x_2x_3^2, x_2^5, x_2^4x_3,
    x_2^3x_3^2, x_2^2x_3^3, x_2x_3^4\rangle
  \end{displaymath}
  is not strongly stable, as
  $x_1\frac{x_1x_2x_3^2}{x_2}=x_1^2x_3^2\notin\I$.  Since
  \begin{displaymath}
    \gin{\I}=\langle x_1^2, x_1x_2^2, x_2^4, x_2^3x_3^2, x_1x_2x_3^3,
    x_2^2x_3^3, x_1x_3^4\rangle\,,
  \end{displaymath}
  $\I$ is not in $\beta$-maximal position, as
  $\beta_4(\I)=(1,3,5)\prec_\lex (1,4,4)=\beta_4(\gin{\I})$
\end{example}

\begin{example}\label{ex:13}
  Let
  $\I=\langle x_1^2,x_1x_2+x_2x_3,x_1x_3,x_2^3,x_2^2x_3\rangle
  \unlhd\kk[x_1,x_2,x_3]$.
  Then
  $\lt{\I_{\langle 2\rangle}}=\langle x_1^2,x_1x_2,x_1x_3,x_2x_3^2\rangle$
  is not quasi-stable, hence $\I$ is not in componentwise quasi-stable
  position. Furthermore,
  $\lt{\I}=\langle x_1^2, x_1x_2, x_1x_3, x_2x_3^2, x_2^3, x_2^2x_3\rangle
  \neq \langle x_1^2, x_1x_2, x_2^2, x_1x_3^2, x_2x_3^2\rangle =\gin{\I}$
  and so we see that $\I$ is not in $\beta$-maximal position, as
  $\beta_2(\I)=(1,1,1)\prec_\lex (1,2,0)=\beta_2(\gin{\I})$.
\end{example}

\begin{example}\label{ex:14}
  Let
  $\I=\langle
  x_1^3,x_2^3,x_1x_3^2,x_2x_3^2,x_3^3\rangle\unlhd\kk[x_1,x_2,x_3]$.
  Then $\I$ is not stable, as $x_1\frac{x_2^3}{x_2}=x_1x_2^2\notin\I$.
  Since
  \begin{displaymath}
    \gin{\I}= \langle x_1^3,x_1^2x_2,x_1x_2^2,
    x_2^3,x_1^2x_3,x_1x_2x_3^2,x_2^2x_3^2, x_1x_3^4, x_2x_3^4, x_3^6
    \rangle\,,
  \end{displaymath}
  $\I$ is not in $\beta$-maximal position, as
  $\beta_3(\I)=(1,1,3)\prec_\lex (1,3,1)=\beta_3(\gin\I)$.
\end{example}

\begin{example}\label{ex:15}
  Let
  $\I=\langle x_1^3,x_1x_2^2+x_2^2x_3,x_2^4,x_1x_3^3,x_3^4,x_2x_3^3\rangle
  \unlhd\kk[x_1,x_2,x_3]$.
  Then
  $\lt\I_{\langle 3\rangle}=\langle x_1x_2^2, x_1^3, x_2^2x_3^3\rangle$ is
  not quasi-stable, hence $\I$ is not in componentwise quasi-stable
  position.
  $\lt{\I}=\langle x_1^3, x_1x_2^2, x_2^4, x_2x_3^3, x_1x_3^3,
  x_3^4\rangle$
  is not stable, as $x_1\frac{x_1x_2^2}{x_2}=x_1^2x_2\notin\lt{\I}$.  With
  \begin{displaymath}
    \gin{\I}=\langle x_1^3, x_1^2x_2, x_1x_2^3, x_2^4, x_1x_2^2x_3,
    x_2^3x_3, x_1^2x_3^2, x_1x_2x_3^3, x_2^2x_3^3, x_1x_3^4, x_2x_3^5,
    x_3^6\rangle\,,
  \end{displaymath}
  we see that $\I$ is in $\beta$-maximal position, since
  \begin{displaymath}
    \begin{array}{rcl}
      \beta_3(\I)= &(1,1,0) & = \beta_3(\gin{\I})\,, \\
      \beta_4(\I)= &(1,4,5) & = \beta_4(\gin{\I})\,, \\
      \beta_5(\I)= &(1,5,13 )& = \beta_5(\gin{\I})\,.
    \end{array}
  \end{displaymath}
\end{example}

\begin{example}\label{ex:16}
  Let
  \begin{displaymath}
    \I=\langle
    x_1^3,x_1x_2^2,x_2^3,x_1^2x_2x_3,x_1^2x_3^2,x_1x_2x_3^2,x_2^2x_3^2,
    x_1x_3^3,x_2x_3^3,x_3^4\rangle\unlhd\kk[x_1,x_2,x_3]\,.
  \end{displaymath}
  Then $\I$ is not stable, as
  $x_1\frac{x_1x_2^2}{x_2}=x_1^2x_2\notin\I$. Since
  \begin{displaymath}
    \gin{\I}=\langle x_1^3, x_1^2x_2, x_1x_2^2, x_2^4, x_2^3x_3, x_1^2x_3^2,
    x_1x_2x_3^2, x_2^2x_3^2, x_1x_3^3, x_2x_3^3, x_3^4\rangle\,,
  \end{displaymath}
  we see that $\I$ is in $\beta$-maximal position, as
  $\beta_3(\I)= (1,2,0) = \beta_3(\gin{\I})$.
\end{example}

\begin{example}\label{ex:17}
  Let
  $\I=\langle
  x_1^3,x_1^2x_2+x_2^3,x_1^2x_3,x_2^4\rangle\unlhd\kk[x_1,x_2,x_3]$.
  Then
  $\lt{\I}=\langle x_1^3,x_1^2x_2,x_1^2x_3,x_1x_2^3,x_2^4,x_2^3x_3\rangle$
  is not strongly stable, as
  $x_1\frac{x_2^3x_3}{x_2}=x_1x_2^2x_3\notin\lt{\I}$. Since
  $\gin{\I}=\langle x_1^3, x_1^2x_2, x_1x_2^2, x_2^4, x_1^2x_3^2\rangle$,
  we see that $\I$ is not in $\beta$-maximal position, as
  $\beta_3(\I)=(1,1,1)\prec_\lex (1,2,0)=\beta_3(\gin{\I})$.
\end{example}

\begin{example}\label{ex:18}
  Let
  $\I=\langle
  x_1^2,x_1x_2+x_2x_3,x_2^3,x_2^2x_3\rangle\unlhd\kk[x_1,x_2,x_3]$.
  Then $\lt{\I_{\langle 2\rangle}}=\langle x_1^2, x_1x_2, x_2x_3^2\rangle$
  is not quasi-stable, hence $\I$ is not in componentwise quasi-stable
  position.  Furthermore,
  $\lt{\I}=\langle x_1^2, x_1x_2, x_2^3, x_2^2x_3, x_2x_3^2\rangle$ is not
  strongly stable, as $x_1\frac{x_2x_3^2}{x_2}=x_1x_3^2\notin\lt{\I}$.
  Since $\gin{\I}=\langle x_1^2, x_1x_2, x_2^3, x_2^2x_3, x_1x_3^2\rangle$,
  we see that $\I$ is in $\beta$-maximal position, as
  $\beta_2(\I)=(1,3,4) = \beta_2(\gin{\I})$.
\end{example}

\begin{example}\label{ex:19}
  Let
  $\I=\langle
  x_1^2,x_1x_2^2,x_2^3,x_2^2x_3^2\rangle\unlhd\kk[x_1,x_2,x_3]$.
  Then $\I$ is not strongly stable, as
  $x_1\frac{x_2^2x_3^2}{x_2}=x_1x_2x_3^2\notin\I$.  Since
  $\gin{\I}=\langle x_1^2, x_1x_2^2, x_2^3, x_1x_2x_3^2\rangle$, we see that
  $\I$ is in $\beta$-maximal position, as
  \begin{displaymath}
    \begin{array}{rcl}
      \beta_2(\I)= &(1,0,0) & = \beta_2(\gin{\I}) \\
      \beta_3(\I)= &(1,3,1) & = \beta_3(\gin{\I})\,.
    \end{array}
  \end{displaymath}
\end{example}

\begin{example}\label{ex:20}
  Let
  $\I=\langle
  x_1^2,x_1x_2+x_2x_3,x_2^3,x_2^2x_3\rangle\unlhd\kk[x_1,x_2,x_3]$.
  Then
  $\lt{\I}=\langle x_1^2,x_1x_2,x_1x_3,x_2^3, x_2^2x_3\rangle \neq \langle
  x_1^2, x_1x_2, x_2^2, x_1x_3^2\rangle =\gin{\I}$
  and so we see that $\I$ is not in $\beta$-maximal position, as
  $\beta_2(\I)=(1,1,1)\prec_\lex (1,2,0)=\beta_2(\gin{\I})$.
\end{example}

\begin{example}\label{ex:21}
  Consider
  $\I=\langle
  x_1^3,x_1^2x_2+x_2x_3^2,x_1x_2^3,x_2^4,x_1x_2^2x_3,x_1^2x_3^2,x_1x_3^4\rangle
  \unlhd{}$ $\kk[x_1,x_2x_3]$.
  Then
  $\lt{\I_{\langle 3\rangle}}=\langle x_1^3, x_1^2x_2, x_1x_2x_3^2,
  x_2x_3^4 \rangle$
  is not quasi-stable, hence $\I$ is not in componentwise quasi-stable
  position.  Furthermore,
  \begin{displaymath}
    \lt{\I}=\langle x_1^3, x_1^2x_2,x_1x_2^3, x_2^4, x_1x_2^2x_3, x_1^2x_3^2,
    x_1x_2x_3^2, x_2^3x_3^2, x_2^2x_3^3, x_1x_3^4, x_2x_3^4 \rangle
  \end{displaymath}
  does not equal
  \begin{displaymath}
    \gin{\I} = \langle x_1^3, x_1^2x_2, x_1x_2^3, x_2^4, x_1x_2^2x_3,
    x_2^3x_3, x_1^2x_3^2, x_1x_2x_3^3, x_2^2x_3^3, x_1x_3^4, x_2x_3^4\rangle\,,
  \end{displaymath}
  but $\I$ is in $\beta$-maximal position, as
  \begin{displaymath}
    \begin{array}{rcl}
      \beta_3(\I)= &(1,1,0) & = \beta_3(\gin{\I})\,, \\
      \beta_4(\I)= &(1,4,5)& = \beta_4(\gin{\I})\,.
    \end{array}
  \end{displaymath}
\end{example}

\begin{example}\label{ex:22}
  Let
  $\K=\langle x_2x_3-x_1x_4,x_1^3-x_2^2x_4,x_2^3-x_1x_3^2\rangle
  \unlhd\kk[x_1,x_2,x_3,x_4]$
  and $\I=\Psi_2\Psi_1(\K)$ with $\Psi_1:(x_3\mapsto x_3+x_1)$ and
  $\Psi_2:(x_2\mapsto x_2+x_1)$. Then
  \begin{displaymath}
    \begin{gathered}
      \lt{\I}=\langle x_1^2, x_1x_2^2, x_2^3, x_1x_2x_3^2, x_1x_3^3,
      x_2^2x_3^3, x_2x_3^4\rangle\neq{}\qquad \\
      \qquad\langle x_1^2, x_1x_2^2, x_2^3,
      x_1x_2x_3^2, x_2^2x_3^2, x_1x_3^4, x_2x_3^4\rangle=\gin{\I}\,,
    \end{gathered}
  \end{displaymath}
  but $\I$ is in $\beta$-maximal position, as
  \begin{displaymath}
    \begin{array}{rcl}
      \beta_2(\I)= &(1,0,0,0) & = \beta_2(\gin{\I})\,, \\
      \beta_3(\I)= &(1,3,1,1) & = \beta_3(\gin{\I})\,, \\
      \beta_4(\I)= &(1,4,7,6)& = \beta_4(\gin{\I})\,.
    \end{array}
  \end{displaymath}
\end{example}

\begin{example}\label{ex:23}
  Let
  $\I=\langle
  x_1^3,x_1^2x_2+x_1x_2x_3,x_1x_2^3,x_1x_2^2x_3,x_1^2x_3^2\rangle\unlhd\kk[x_1,x_2,x_3]$.
  Then
  $\lt{\I_{\langle 3\rangle}}=\langle x_1^3,x_1^2x_2,x_1x_2x_3^2\rangle$ is
  not quasi-stable, hence $\I$ is not in componentwise quasi-stable
  position.  Furthermore,
  \begin{displaymath}
    \lt{\I}=\langle x_1^3, x_1^2x_2, x_1x_2^3, x_1x_2^2x_3, x_1^2x_3^2,
    x_1x_2x_3^2 \rangle=\gin{\I}
  \end{displaymath}
  and so $\I$ is in $\beta$-maximal position.
\end{example}

\begin{example}\label{ex:24}
  The final ideal that is in any position is simply
  $\langle x_1\rangle\unlhd\kk[x_1]$.
\end{example}

\section{A Deterministic Algorithm for Stable Positions}
\label{sec:algo}

\subsection{Description of the Algorithm}
\label{sec:algodesc}

We discuss now the main computational result of this article: a
deterministic algorithm that for a coefficient field of characteristic zero
incrementally transforms into any of the generic positions related to
stability\footnote{With the help of the criterion of Proposition
  \ref{prop:cqs}, this also includes componentwise quasi-stability---see
  Remark \ref{rem:cqsalg} below for more details.} and for a field of
positive characteristic $p$ into any of the corresponding $p$-variants.  It
performs at each step an \emph{elementary move}, i.\,e.\ for a single pair
$(k,\ell)$ of indices with $\ell<k$ we transform
$x_{k}\mapsto x_{k}+x_{\ell}$ with all other variables unchanged, so that
we obtain a fairly sparse transformation if not too many steps are
necessary.  Such a move transforms any term $x^{\mu}$ containing $x_{k}$
into a linear combination of terms of which $x^{\mu}$ is the smallest with
respect to the degree reverse lexicographic order (for this reason it is crucial
that this order is used).  While the algorithm itself is thus fairly
simple, it turns out that quite some work is required to prove that it
always terminates after a finite number of transformations.

The termination proof is based on the following simple observation.  We
proceed as in the above discussion of a deterministic way to compute
$\gin{\I}$: a linear coordinate transformation with undetermined
coefficients is performed and then a Gr\"obner system is computed with the
coefficients as parameters.  By Remark \ref{rem:finlt}, any ideal possesses
only finitely many different leading ideals under arbitrary linear
transformations.  We define now an ordering on the set of these leading
ideals and then show that our algorithm produces a strictly ascending
sequence of leading ideals.  Obviously, this implies termination.

\begin{definition}\label{def:ls}
  Let $F\subset\P$ be a finite set of polynomials with leading terms
  $\lt{F}=\{t_1,\ldots,t_\ell\}$ such that
  $t_1\succ_\rl \cdots \succ_\rl t_\ell$ where now $\prec_{\rl}$ denotes
  the pure reverse lexicographic order.\footnote{Note that opposed to the
    \emph{degree} reverse lexicographic order, $\prec_{\rl}$ is not a term
    order.  Since we are, however, exclusively considering homogeneous
    polynomials, we may always pretend that the leading term has been
    selected via $\prec_{\rl}$.}  Then we denote the ordered tuple of these
  leading terms by $\ls{(F)} = (t_1,\ldots,t_\ell)$.  If
  $F,\tilde{F}\subset\P$ are two finite sets of polynomials with
  $\ls{(F)}=(t_1,\ldots,t_\ell)$ and
  $\ls{(\tilde{F})}=(\tilde{t}_1,\ldots,\tilde{t}_{\tilde{\ell}})$, then we
  define an ordering by setting
  \begin{displaymath}
    \ls(F) \prec_\ls \ls(\tilde{F})  \ \ \Longleftrightarrow \  \  
    \left\{ \begin{array}{lr}
             \exists\,j \leq \min{(\ell,\tilde{\ell})}\ \forall\,i < j :
             t_i=\tilde{t}_i \wedge  t_j \prec_{\rl} \tilde{t}_j  & 
             \text{or}  \\
             \forall\,j \leq \min{(\ell,\tilde{\ell})} : 
             t_j=\tilde{t}_j \wedge \ell < \tilde{\ell}\,. & 
           \end{array}\right.
  \end{displaymath}
\end{definition}

For notational simplicity, we present our Algorithm \ref{algo:sstrafo} for
the special case of strongly stable position.  If the algorithm terminates,
then its correctness is obvious, as the condition in Line
\ref{line:sstrafo} just encodes the definition of a strongly stable ideal.
The only not so obvious part of the algorithm is the \texttt{while} loop in
Line \ref{line:wlbc}.  It will become later evident why we need it.  In
fact, it only works, if $\ch{\kk}=0$.  We will discuss later the
modifications required for positive characteristic.

\begin{algorithm}[ht]
  \begin{algorithmic}[1]
    \REQUIRE reduced Gr\"obner basis $G$ of homogeneous ideal $\I\unlhd\P$ 
    \ENSURE  a linear change of coordinates $\Psi$ such that
             $\lt{\Psi(\I)}$ is strongly stable 
    \STATE $\Psi:=\mathrm{id}$;
    \WHILE {$\exists\, g\in G,\ 1\leq j\leq n,\ 1\leq i < j\,:\,
                  x_j\mid\lt{g} \wedge  x_{i}\frac{\lt{g}}{x_{j}}
                  \notin\langle\lt{G}\rangle$}\label{line:sstrafo}  
        \STATE $\psi:=(x_{j}\mapsto x_{j}+x_{i})$; $\Psi=\psi\circ\Psi$
        \STATE $\tilde{G}:=
            \text{\textsc{ReducedGr\"obnerBasis}}\bigl(\psi(G)\bigr)$
            \label{line:sstrafo:gb}
        \WHILE {$\ls(G)\succeq_\ls \ls(\tilde{G})$}\label{line:wlbc}
            \STATE $\psi:=(x_{j}\mapsto x_{j}+x_{i})$; 
                   $\Psi=\psi\circ\Psi$
	     \STATE $\tilde{G}:=
                     \text{\textsc{ReducedGr\"obnerBasis}}
                     \bigl(\psi(\tilde{G})\bigr)$\label{line:sstrafo:gb2}
        \ENDWHILE
         \STATE $G:=\tilde{G}$
    \ENDWHILE
    \RETURN $\Psi$
  \end{algorithmic}
  \caption{\textsc{SS-Trafo}: Transformation to strongly stable
    position}\label{algo:sstrafo} 
\end{algorithm}

To apply the algorithm for a different notion of stability, one only has to
modify the condition in Line 2 so that it encodes the corresponding
stability criterion.  Then again the correctness is obvious and the precise
nature of the stability criterion will play no role in the termination
proof below.

\begin{example}
  Let $\I=\langle x_1^3,x_2^3,x_2^2x_3\rangle \unlhd \kk[x_1,x_2,x_3]$.
  $\I$ is not strongly stable, as
  $x_1\frac{x_2^2x_3}{x_3}=x_1x_2^2\notin \I$.  We perform the coordinate
  transformation $\Psi_1:(x_3\mapsto x_3+x_1)$ and obtain
  \begin{displaymath}
    \lt{\Psi_1(\I)}=\langle x_1^3,x_1x_2^2,x_2^3,x_2^2x_3^3\rangle\,.
  \end{displaymath}
  Since $(x_1^3,x_2^3,x_2^2x_3)\prec_\ls(x_1^3,x_1x_2^2,x_2^3,x_2^2x_3^2)$,
  we do not enter the \texttt{while} loop in Line \ref{line:wlbc}.  But
  $\lt{\Psi_1(\I)}$ is still not strongly stable, as
  $x_1\frac{x_1x_2^2}{x_2}=x_1^2x_2\notin \lt{(\Psi_1(\I))}$.  Thus we
  perform as second coordinate transformation $\Psi_2:(x_2\mapsto x_2+x_1)$
  leading to
  \begin{displaymath}
    \lt{\Psi_2\bigl(\Psi_1(\I)\bigr)}=
    \langle x_1^3,x_1^2x_2,x_1x_2^2,x_2^4,x_1^2x_3^3\rangle\,.
  \end{displaymath}
  Again we do not enter the inner \texttt{while} loop, as this time
  $(x_1^3,x_1x_2^2,x_2^3,x_2^2x_3^2)\prec_\ls(x_1^3,x_1^2x_2,x_1x_2^2,x_2^4,x_1^2x_3^3)$.
  Now there are no obstructions left, i.\,e.\
  $\lt{\Psi_2\bigl(\Psi_1(\I)\bigr)}$ is strongly stable (in this case we
  even have $\lt{\Psi_2\bigl(\Psi_1(\I)\bigr)}=\gin\I$).
\end{example}

The next example shows explicitly that the result of Algorithm
\ref{algo:sstrafo} is not unique.  More precisely, in the outer
\texttt{while} loop one finds generally more than one obstruction $(i,j)$
and each choice will lead to a different transformations.

\begin{example}
  Let $\I=\langle x_1^2,x_1x_2,x_2x_3,x_2^3\rangle\unlhd\kk[x_1,x_2,x_3]$.
  Since both $x_1\frac{x_2x_3}{x_2}=x_1x_3$ and
  $x_2\frac{x_2x_3}{x_3}=x_2^2$ are not contained in $\I$, we have the
  choice to perform either $\Psi_1:(x_2\mapsto x_2+x_1)$ or
  $\Psi_2:(x_3\mapsto x_3+x_1)$.  Since
  \begin{displaymath}
    \renewcommand{\arraystretch}{1.5}
    \begin{array}{rcl}
      \lt{\Psi_1(\I)} & = & 
          \langle x_1^2,x_1x_2,x_1x_3, x_2^3,x_2^2x_3\rangle \\
      \lt{\Psi_2(\I)} & = & \langle x_1^2,x_1x_2,x_2^3,x_2x_{3}\rangle\,,
    \end{array}
  \end{displaymath}	
  we see that applying $\Psi_1$ directly leads to a strongly stable ideal
  whereas $\lt{\Psi_2(\I)}$ is still not strongly stable
  ($x_1\frac{x_2x_3^2}{x_2}=x_1x_3^2$ is not contained).  But
  \begin{displaymath}
    \lt{\Psi_1\bigl(\Psi_2(\I)\bigr)}=
    \langle x_1^2, x_1x_2, x_2^2, x_1x_3^2\rangle
  \end{displaymath}
  is strongly stable and not equal to $\lt{\Psi_1(\I)}$.
\end{example}

\begin{remark}
  Although in this article we are only concerned with the principal
  question of deterministically obtaining generic positions, we want to
  comment briefly on some efficiency issues.  In a concrete computer
  realisation of Algorithm \ref{algo:sstrafo}, any optimisation will aim at
  reducing either the number of checks for obstructions or the total number
  of transformations.  One can think of quite a number of natural
  strategies to achieve these goals.  However, for each of them one can
  provide counter examples \citep[Sect.~2.3]{ms:diss}, so that none of them
  is always successful.  The relative merits of these strategies can thus
  be assessed only in extensive benchmarks.

  We consider here only one particularly natural strategy, namely to attack
  always the obstructions of lowest degree.  The logic behind this strategy
  is the expectation that no transformation introduces obstructions in
  lower degrees and thus that each degree must be considered only once.
  However, this expectation is wrong, as the following example
  demonstrates.  Consider the ideal
  $\I=\langle x_1^3,x_1^2x_2+x_2^3,x^2x_3 \rangle\unlhd\kk[x_1,x_2,x_3]$
  with leading ideal
  \begin{displaymath}
    \lt{\I}=\langle x_1^3, x_1^2x_2,x_1^2x_3, x_1x_2^3, x_2^3x_3, x_2^5 \rangle\,.
  \end{displaymath}
  There are no obstructions in degree 3, which is the lowest degree of a
  generator.  But since $x_2\frac{x_2^3x_3}{x_3}=x_2^4\notin\lt{\I}$, there
  exists one in degree 4.  We can remove it by applying the transformation
  $\Psi:(x_3\mapsto x_3+x_2)$.  The new leading ideal
  \begin{displaymath}
    \lt{\Psi(\I)}=\langle x_1^3, x_1^2x_2, x_2^3, x_1^2x_3^3 \rangle
  \end{displaymath}
  has no obstructions in degree 4 or 5, which is the highest degree of a
  minimal generator.  But $\lt{\Psi(\I)}$ is not strongly stable, since now
  an obstruction appears in degree 3:
  $x_1\frac{x_2^3}{x_2}=x_1x_2^2\notin\lt{\Psi(\I)}$.
\end{remark}

\begin{remark}\label{rem:cqsalg}
  The definition of a componentwise quasi-stable position is quite
  different from the one of a quasi-stable position, as it uses the
  component ideals $\I_{\lspan{d}}$ (which are truly polynomial) instead of
  the monomial ideal $\lt{\I}$.  Thus a straightforward algorithm for
  obtaining a componentwise quasi-stable position would analyse all these
  ideals simultaneously which is very expensive.  Our results in Section
  \ref{sec:compstab} allow us to modify Algorithm \ref{algo:sstrafo} in
  such a way that it can be directly used for this task.

  First of all, we use the obvious variant of Algorithm \ref{algo:sstrafo}
  to put $\I$ into a quasi-stable position.  Then we start Algorithm
  \ref{algo:sstrafo} again with the condition in Line \ref{line:sstrafo}
  replaced by the sufficient criterion derived in Proposition
  \ref{prop:cqs}.  The implementation of this criterion requires two
  further modifications: Instead of reduced Gr\"obner bases we compute
  Pommaret bases in the Lines \ref{line:sstrafo:gb} and
  \ref{line:sstrafo:gb2} (their finiteness is ensured, as we are in a
  quasi-stable position) and this computation must be performed in such a
  way that we also obtain all the syzygies corresponding to the involutive
  standard representations \eqref{eq:isr}.

  As already mentioned in Remark \ref{rem:css}, we can similarly transform
  into a componentwise (strongly) stable position.  We only have to put
  $\I$ in the first step into a (strongly) stable position.  Then we can
  use the same modified algorithm as for a componentwise quasi-stable
  position.
\end{remark}

\subsection{The Termination Proof}
\label{sec:termproof}

Let $F=\{f_1,\ldots,f_\ell\}\subset\P$ be a finite set of polynomials.  We
call $F$ \emph{completely autoreduced}, if no term contained in the support
of a polynomial $f_{i}$ is divisible by a leading term $\lt{f_{j}}$ with
$j\neq i$.  $F$ is \emph{head autoreduced}, if no leading term $\lt{f_{i}}$
is divisible by another leading term $\lt{f_{j}}$.  By an obvious
algorithm, any set $F$ can be rendered either completely or head
autoreduced.  We denote the results by $F^\blacktriangle$ and by
$F^\triangle$, respectively.  Furthermore, if $0\neq f\in\P$ is an
arbitrary non-vanishing polynomial and $t\in\supp{(f)}$ a term appearing in
it, then we denote the coefficient of $t$ in $f$ by $\c_f(t)$.

\begin{lemma}\label{lem:step2}
  Let $F\subseteq\P$ be a completely autoreduced set of polynomials.  Let
  $\Psi:(x_j \mapsto x_j+ax_i)$ be a linear coordinate transformation with
  $i<j$ and a parameter $a\in\kk^{\times}$.  If the field $\kk$ possesses
  more than $2\deg{F}$ elements, then there exists a value $a$ such that
  \begin{displaymath}
    \ls(F)\preceq_\ls\ls\bigl(\Psi(F)^\triangle\bigr)\,.
  \end{displaymath}
  If $\kk$ is an infinite field, then this inequality will hold for any
  (Zariski) generic choice of the parameter $a$.
\end{lemma}

\begin{proof}
  We order $F=\{f_1,\ldots,f_\ell\}$ such that
  $\lt{f_k} \prec_{\rl} \lt{f_l}$ whenever $k>l$. Furthermore we set
  $t_k=\lt{f_k}$ and $s_k=\lt{\Psi(f_k)}$ for each $k$.  Without loss of
  generality, we assume that $\lc{f_k}=1$ for each $k$.  It is easy to see
  that $ t_k \preceq_{\rl} s_k$ for all $k$, as $i<j$. If $t_k=s_k$ for all
  $k$, then there is nothing to prove, since then
  $\lt{F}=\lt{\Psi(F)}=\lt{\Psi(F)^\Delta}$.  Otherwise let $\alpha$ be the
  smallest index such that $t_\alpha\neq s_\alpha$.  In other words:
  $t_k=s_k$ for all $k < \alpha$, $t_\alpha \prec_{\rl} s_\alpha$ and
  $t_k \preceq_{\rl} s_k$ for all $k > \alpha$.  Let $h_\alpha$ be the
  remainder of $\Psi(f_\alpha)$ after reducing it by the set
  $\bigl\{\Psi(f_1),\ldots ,\Psi(f_{\alpha-1})\bigr\}$ --- note that this
  set is head but in general not completely autoreduced.  We want to show
  that $t_\alpha \in \supp{(h_\alpha)}$, as then obviously
  $\lt{h_\alpha} \succeq_{\rl} t_\alpha$.

  If $h_\alpha=\Psi(f_\alpha)$, we are done, since then
  $t_\alpha\in\supp{\bigl(\Psi(f_\alpha)\bigr)}$.  Otherwise there exists
  an index $\beta<\alpha$ such that $s_\beta=t_\beta$ divides $s_\alpha$.
  So the question arises whether or not $t_\alpha$ remains in the support
  of
  \begin{displaymath}
    h_\beta=\Psi(f_\alpha)-
         \dfrac{\c_{\Psi(f_\alpha)}(s_\alpha)s_\alpha}
               {\c_{\Psi(f_\beta)}(t_\beta)t_\beta}\Psi(f_{\beta})\,.
  \end{displaymath}

  Let us assume that this was not the case.  Hence in $\Psi(f_\beta)$ a
  monomial $m_\beta=\c_{\Psi(f_\beta)}(t_{m_\beta})t_{m_\beta}$ exists
  which causes the cancellation of $t_\alpha$.  Clearing denominators, we
  arrive thus at the equality
  \begin{equation}\label{eq:tacancel}
    \c_{\Psi(f_\alpha)}(t_\alpha)\c_{\Psi(f_\beta)}(t_\beta)t_\alpha t_\beta =
    \c_{\Psi(f_\alpha)}(s_\alpha) \c_{\Psi(f_\beta)}(t_{m_\beta})s_\alpha 
        t_{m_\beta}\;.
  \end{equation}
  We analyse now the appearing coefficients as elements of $\kk[a]$,
  i.\,e.\ as polynomials in the parameter $a$.  Because of the form of the
  transformation $\Psi$, the term $1$ is contained in both
  $\supp{\bigl(\c_{\Psi(f_\alpha)}(t_\alpha)\bigr)}$ and
  $\supp{\bigl(\c_{\Psi(f_\beta)}(t_\beta)\bigr)}$ and hence also in
  $\supp{\bigl(\c_{\Psi(f_\alpha)}(t_\alpha)\c_{\Psi(f_\beta)}(t_\beta)\bigr)}$.
  But our assumption $s_\alpha\succ_{\rl} t_\alpha$ implies that
  $1\notin \supp{\bigl(\c_{\Psi(f_\alpha)}(s_\alpha)\bigr)}$ and thus
  $1 \notin \supp{\bigl(\c_{\Psi(f_\alpha)}(s_\alpha)
    \c_{\Psi(f_\beta)}(t_{m_\beta})\bigr)}$.
  This argument shows that as polynomials in $a$ the two decisive
  coefficients $\c_{\Psi(f_\alpha)}(t_\alpha)\c_{\Psi(f_\beta)}(t_\beta)$
  and $\c_{\Psi(f_\alpha)}(s_\alpha)\c_{\Psi(f_\beta)}(t_{m_\beta})$ cannot
  be equal.  For any value $a$ outside the set
  \begin{displaymath}
    \V\bigl(\c_{\Psi(f_\alpha)}(t_\alpha)\c_{\Psi(f_\beta)}(t_\beta) - 
            \c_{\Psi(f_\alpha)}(s_\alpha)\c_{\Psi(f_\beta)}(t_{m_\beta})\bigr)
    \subseteq\kk 
  \end{displaymath}
  therefore the equality (\ref{eq:tacancel}) cannot hold which contradicts
  our assumption that $t_\alpha\notin \supp{(h_\beta)}$.  Each coefficient
  in $\Psi(f_\alpha)$ is a polynomial in $a$ with its degree bounded by
  $\deg{f_\alpha}$ and analogously for $\Psi(f_\beta)$.  Thus there are at
  most $2\deg{F}$ ``bad'' values $a$ and for a sufficiently large field
  $\kk$ we can always find a ``good'' one.

  Clearing denominators in the equation for the coefficient of $t_\alpha$
  in $h_\beta$, we obtain the equality
  \begin{displaymath}
    \c_{\Psi(f_\beta)}(t_\beta)\c_{h_\beta}(t_\alpha) = 
    \c_{\Psi(f_\beta)}(t_\beta)\c_{\Psi(f_\alpha)}(t_\alpha)-
    \c_{\Psi(f_\alpha)}(s_\alpha)\c_{\Psi(f_\beta)}(t_{m_\beta})\,.
  \end{displaymath}
  With the arguments from above, we find
  $1\in\supp{\bigl(\c_{\Psi(f_\beta)}(t_\beta)\c_{h_\beta}(t_\alpha)\bigr)}$
  and thus
  \begin{equation}\label{eq:supp}
    1 \in \supp{\bigl(\c_{h_\beta}(t_\alpha)\bigr)}\,.
  \end{equation}
  If already $h_\beta=h_\alpha$, we are done.  Otherwise there exists an
  index $\gamma<\alpha$ such that $s_\gamma=t_\gamma$ divides
  $\lt{h_\beta}=t_{h_\beta}$.  The existence of such a divisor shows that
  $t_{h_\beta}$ cannot be equal to $t_\alpha$ since $F$ is a completely
  autoreduced set---note that we could not argue like this if $F$ was
  only head autoreduced---and therefore
  \begin{equation}\label{eq:thvgreater}
    t_{h_\beta} \succ_{\rl} t_\alpha\,.
  \end{equation}
  As above we must show that $t_\alpha$ remains in the support of
  \begin{displaymath}
    h_\gamma=h_\beta-
        \frac{\c_{h_\beta}(t_{h_\beta})t_{h_\beta}}
             {\c_{\Psi(f_\gamma)}(t_\gamma)t_\gamma} \Psi(f_\gamma)\,.
  \end{displaymath}
  Let us assume that this was not the case.  Hence in $\Psi(f_\gamma)$ a
  monomial $m_\gamma=\c_{\Psi(f_\gamma)}(t_{m_\gamma})t_{m_\gamma}$ exists
  such that---after clearing denominators---
  \begin{equation}\label{eq:tucancel2}
    \c_{h_\beta}(t_\alpha)\c_{\Psi(f_\gamma)}(t_\gamma)t_\alpha t_\gamma = 
    \c_{h_\beta}(t_{h_\beta})\c_{\Psi(f_\gamma)}(t_{m_\gamma})
        t_{h_\beta}t_{m_\gamma}\,.
  \end{equation}
  Let us again analyse the coefficients.  As above, we immediately find
  that $1\in\supp{\bigl(\c_{\Psi(f_\gamma)}(t_\gamma)\bigr)}$ because of
  the form of the transformation $\Psi$.  In \eqref{eq:supp} we already saw
  that $1\in\supp{\bigl(\c_{h_\beta}(t_\alpha)\bigr)}$, hence
  $1\in\supp{\bigl(\c_{h_\beta}(t_\alpha)\c_{\Psi(f_\gamma)}(t_\gamma)\bigr)}$.
  We are done, if we are able to show that
  \begin{equation}\label{eq:laststep}
    1\notin \supp{\bigl(\c_{h_\beta}(t_{h_\beta})\bigr)}\,,
  \end{equation}  
  as then
  $1\notin \supp{\bigl(\c_{h_\beta}(t_{h_\beta})
    \c_{\Psi(f_\gamma)}(t_{m_\gamma})\bigr)}$
  and so again the equality \eqref{eq:tucancel2} cannot hold for all values
  $a$ in a sufficiently large field $\kk$.

  To show \eqref{eq:laststep}, we recall the construction of $h_\beta$, 
  \begin{displaymath}
    h_\beta=\Psi(f_\alpha)-
        \dfrac{\c_{\Psi(f_\alpha)}(s_\alpha)s_\alpha}
              {\c_{\Psi(f_\beta)}(t_\beta)t_\beta}\Psi(f_{\beta})\,,
  \end{displaymath}
  which implies the equality
  \begin{equation}\label{eg:thb}
    \c_{h_\beta}(t_{h_\beta})\c_{\Psi(f_\beta)}(t_\beta) = 
    \c_{\Psi(f_\alpha)}(t_{h_\beta})\c_{\Psi(f_\beta)}(t_\beta)-
    \c_{\Psi(f_\alpha)}(s_a)\c_{\Psi(f_\beta)}(t_{h_\beta})\,.
  \end{equation}
  On one hand we note that
  $1\notin\supp{\bigl(\c_{\Psi(f_\alpha)}(t)\bigr)}$ for all terms
  $t\in\supp{\bigl(\Psi(f_\alpha)\bigr)}$ with $t\succ_{\rl} t_\alpha$.
  Thus, since $t_{h_\beta} \succ_{\rl} t_\alpha$ by \eqref{eq:thvgreater},
  it follows that if $t_{h_\beta}\in\supp{\bigl(\Psi(f_\alpha)\bigr)}$,
  then $1\notin \supp{\bigl(\c_{\Psi(f_\alpha)}(t_{h_\beta})\bigr)}$ and
  therefore
  \begin{displaymath}
    1\notin\supp{\bigl(\c_{\Psi(f_\alpha)}(t_{h_\beta})
                       \c_{\Psi(f_\beta)}(t_\beta)\bigr)}\,.
  \end{displaymath}
  On the other hand,
  $1\notin\supp{\bigl(\c_{\Psi(f_\alpha)}(s_\alpha)\bigr)}$ as we have seen
  above and so
  \begin{displaymath}
    1\notin\supp{\bigl(\c_{\Psi(f_\alpha)}(s_a)
                       \c_{\Psi(f_\beta)}(t_{h_\beta})\bigr)}\,.
  \end{displaymath}
  Since at least one of the coefficients $\c_{\Psi(f_\alpha)}(t_{h_\beta})$
  and $\c_{\Psi(f_\beta)}(t_{h_\beta})$ must be nonzero, we conclude from
  \eqref{eg:thb} that
  $1\notin\supp{\bigl(\c_{h_\beta}(t_{h_\beta})
    \c_{\Psi(f_\beta)}(t_\beta)\bigr)}$.
  Now \eqref{eq:laststep} follows from the fact that
  $1\in\supp{\bigl(\c_{\Psi(f_\beta)}(t_\beta)\bigr)}$. 

  We can repeat this procedure for each reduction step until we end up at
  the final result $h_\alpha$ and the arguments imply then that
  $t_\alpha\in \supp{(h_\alpha)}$.  Hence either
  $t_\alpha\prec_{\rl} t_{h_\alpha}$ or $t_\alpha=t_{h_\alpha}$.  Let us
  first assume that $t_\alpha\prec_{\rl} t_{h_\alpha}$.  It is not clear
  that the set $\bigl\{\Psi(f_1),\ldots,\Psi(f_{\alpha-1}),h_\alpha\bigr\}$
  is head autoreduced, as it could happen that there is an index
  $\delta<\alpha$ such that $t_{h_\alpha}$ divides $s_\delta=t_\delta$.
  Since $t_{h_\alpha}\neq t_\delta$ by the construction of $h_\alpha$, we
  know that $t_{h_\alpha} \succ_{\rl} t_\delta$.  In this case we check
  whether or not the set
  $\bigl\{\Psi(f_1),\ldots,\Psi(f_{\delta-1}),h_\alpha\bigr\}$ is head
  autoreduced. If it is not, then there is an index $\epsilon<\delta$ such
  that $t_{h_\alpha}$ divides $s_\epsilon=t_\epsilon$ and we check again
  whether or not the set
  $\bigl\{\Psi(f_1),\ldots,\Psi(f_{\epsilon-1}),h_\alpha\bigr\}$ is head
  autoreduced.  We continue like this until we reach an index
  $\zeta<\epsilon$ such that the set
  $\bigl\{\Psi(f_1),\ldots,\Psi(f_{\zeta-1}),h_\alpha\bigr\}$ is head
  autoreduced.  It is still not clear whether this set is a subset of
  $\Psi(F)^\Delta$, but we can see that
  $\lt{f_\zeta} \prec_{\rl} \lt{h_\alpha}$ and thus
  \begin{displaymath}
    \ls(f_1,\ldots,f_\zeta) \prec_\ls 
    \ls\bigl(\Psi(f_1),\ldots,\Psi(f_{\zeta-1}),h_\alpha\bigr)\,.
  \end{displaymath}
  If $\Psi(F)^\Delta=\{\hat{f}_1,\ldots,\hat{f}_{\hat{m}} \}$, then of
  course
  \begin{displaymath}
    \ls\bigl(\Psi(f_1),\ldots,\Psi(f_{\zeta-1}),h_\alpha\bigr) 
    \preceq_\ls \ls(\hat{f}_1,\ldots,\hat{f}_\zeta)
  \end{displaymath}
  and this inequality suffices to prove our claim $\ls(F)\prec_\ls
  \ls(\Psi(F)^\Delta)$.

  There remains the case $t_\alpha=t_{h_\alpha}$.  Now we have to look for
  the smallest index $\alpha^\prime >\alpha$ such that
  $t_{\alpha^\prime}\neq s_{\alpha^\prime}$.  Then we reduce
  $\Psi(f_{\alpha^\prime})$ by the set
  \begin{equation}\label{eq:set}
    \bigl\{\Psi(f_1),\ldots,\Psi(f_{\alpha-1}),h_\alpha,
           \Psi(f_{\alpha+1}),\ldots,\Psi(f_{\alpha^\prime-1})\bigr\}
  \end{equation}
  to the polynomial $h_{\alpha^\prime}$ in the same way as above --- note
  that \eqref{eq:set} is head autoreduced since the leading terms did not
  change in comparison to the completely autoreduced set $F$. It is clear
  that if we go on like this, then we will either end up at
  $\Psi(F)^\Delta$ with $\lt{\hat{f}_k}= \lt{f_k}$ for all $k$ which would
  mean that $\ls(F) = \ls(\Psi(F)^\Delta)$ or we find a generator
  $h_\omega$ with $t_\omega \prec_{\rl} t_{h_\omega}$ which finishes our
  proof.\qed 
\end{proof}

\begin{lemma}\label{lem:step1}
  Let $\I\unlhd \P$ be an ideal and $G$ its reduced Gr\"obner basis.  Let
  $\Psi:(x_j \mapsto x_j+ax_i)$ be a linear coordinate transformation with
  $i<j$ and a parameter $a\in\kk^{\times}$.  Furthermore, let $\tilde{G}$
  be the reduced Gr\"obner basis of the transformed ideal $\Psi(\I)$.  Then
  \begin{displaymath}
    \ls\bigl(\Psi(G)^\Delta\bigr) \preceq_\ls \ls(\tilde{G})\,.
  \end{displaymath}
\end{lemma}

\begin{proof}
  Suppose that $\ls(\Psi(G)^\Delta)=(t_1,\ldots,t_\ell)$ and
  $\ls(\tilde{G})=(\tilde{t}_1,\ldots,\tilde{t}_{\tilde{\ell}})$.  By
  definition of a Gr\"obner basis, there exists for any leading term
  $t_k\in\lt{\bigl(\Psi(G)^\Delta\bigr)} \subseteq \lt{\langle
    \Psi(G)^\Delta\rangle} = \langle\lt{\tilde{G}}\rangle$ a generator
  $\tilde{g}_k\in \tilde{G}$ such that $\lt{\tilde{g}_k}$ divides $t_k$ and
  therefore\footnote{As $\prec_{\rl}$ is not a term order, it shows a quite
    different behaviour compared to the partial order defined by
    divisibility: $s\mid t$ trivially implies $s\succeq_{\rl}t$.}
  $\lt{\tilde{g}_k}\succeq_{\rl} t_k$.  Now we compare the two lists
  beginning with the first entry.

  Let $\lt{\tilde{g}_1} = \tilde{t}_\alpha$.  If $\alpha>1$, we are done
  since then
  $\tilde{t}_1\succ_{\rl} \tilde{t}_\alpha= \lt{\tilde{g}_1} \succeq_{\rl}
  t_1$.
  So we assume $\lt{\tilde{g}_1}=\tilde{t}_1$.  We are again done, if
  $\tilde{t}_1\succ_{\rl}t_1$.  Thus we further assume that
  $t_1=\tilde{t}_1$ and go on with the next entry.  We note that
  $\tilde{g}_1\neq\tilde{g}_2$, since otherwise
  $t_1=\tilde{t}_1=\lt{\tilde{g}_1}=\lt{\tilde{g}_2}$ divides $t_2$ which
  contradicts $\Psi(G)^\Delta$ being head autoreduced.  Now we have to
  check which position $\lt{\tilde{g}_2}=\tilde{t}_\beta$ has in the list
  $\ls(\tilde{G})$.  Since $\tilde{G}$ is reduced
  $\lt{\tilde{g}_1} \neq \lt{\tilde{g}_2}$ and therefore $\beta>1$.  If
  $\beta>2$, we again have the situation
  $\tilde{t}_2\succ_{\rl}\tilde{t}_\beta=\lt{(\tilde{g}_2)}\succeq_{\rl}
  t_2$
  and are done.  Otherwise $\beta=2$ and so either
  $\tilde{t}_2\succ_{\rl}t_2$ or $\tilde{t}_2=t_2$.  In the first case, our
  assertion follows and in the second one we go on with the next
  entry. Thus sooner or later we either find an index $\omega$ with
  $\tilde{t}_\omega \succ_\rl t_\omega$ which shows that
  $\ls(\Psi(G)^\Delta)\prec_\ls\ls(\tilde{G})$ or
  \begin{equation}\label{eq:equalcase}
    \tilde{t}_k=t_k \mbox{ for all }k \leq \min(\tilde{\ell},\ell)
  \end{equation}
  Assuming that \eqref{eq:equalcase} holds, we note that since $\tilde{G}$
  is a Gr\"obner basis of $\langle \Psi(G)^\Delta\rangle$ and both
  $\Psi(G)^\Delta$ and $\tilde{G}$ are reduced sets we must have
  $\ell \leq \tilde{\ell}$.  Hence it follows that
  $\ls(\Psi(G)^\Delta)=\ls(\tilde{G})$ if $\ell=\tilde{\ell}$ and
  $\ls(\Psi(G)^\Delta) \prec_\ls \ls(\tilde{G})$ if
  $\ell < \tilde{\ell}$.\qed
\end{proof}

The next, rather elementary lemma studies the effect of our basic
coordinate transformations on a polynomial.  It encapsulates the dependence
of our approach on the characteristic of the base field $\kk$ and shows why
for a positive characteristic in general only the $p$-version of our
stability notions are reachable: some terms simply cannot be generated by
linear coordinate transformations.

\begin{lemma}\label{lem:obstrafo}
  Let $f\in\P\setminus\kk$ be a non-constant polynomial and
  $\Psi:(x_j \mapsto x_j+ax_i)$ a linear coordinate transformation with
  $i<j$ and a parameter $a\in\kk^{\times}$.  Furthermore, let
  $\xv^\mu\in\supp(f)$ be a term in the support of $f$ with $\mu_j>0$.  If
  $\ch{\kk}=0$, then, for a generic choice of $a$, all terms of the form
  $x_{i}^{\mu_j-s}\xv^{\mu}/x_j^{\mu_j-s}$ with $1\leq s\leq \mu_{j}$
  appear in the support of $\Psi(f)$.  If $\ch{\kk}=p>0$ and $\kk$ has more
  then $\deg{f}$ elements, then for each term of this form with
  $s\prec_{p}\mu_{j}$ at least one value of $a$ exists such that the term
  appears in $\supp{\bigl(\Psi(f)\bigr)}$.
\end{lemma}

\begin{proof}
  An arbitrary term $\xv^{\nu}\in\supp{f}$ is transformed into the
  polynomial 
  \begin{equation}\label{eq:termtrafo}
    \Psi(\xv^{\nu})=
        \sum_{s=0}^{\nu_{j}}\binom{\nu_j}{s}a^{\nu_j-s}x_i^{\nu_j-s}
                          \dfrac{\xv^\nu}{x_j^{\nu_j-s}}\,.
  \end{equation}
  Thus all terms in the transformed polynomial $\Psi(f)$ have as
  coefficients polynomials in $\kk[a]$ of degree at most $\deg{f}$.  Now we
  analyse the coefficients of the terms
  $x_{i}^{\mu_j-s}\xv^{\mu}/x_j^{\mu_j-s}$.  Each of these terms appears in
  $\Psi(x^{\mu})$ with coefficient $a^{\mu_j-s}$.  If such a term also
  appears in $\Psi(\xv^{\nu})$ with $\nu\neq\mu$, then the exponent vector
  $\nu$ must satisfy $\nu_{k}=\mu_{k}$ for all $k\neq i,j$ and
  $\nu_{i}+\nu_{j}=\mu_{i}+\mu_{j}$.  This implies that the coefficient
  $a^{\nu_{j}-s}$ is different from $a^{\mu_j-s}$.  Hence none of the terms
  we consider has a zero polynomial as coefficient.  It is now
  straightforward to verify our assertion.\qed
\end{proof}

\begin{proposition}\label{prop:list}
  Let $\I\unlhd\P$ be an ideal and $G$ its reduced Gr\"obner basis.  Assume
  that for a generator $g\in G$ with $\lt{g}=\xv^\mu$ there exist indices
  $i,j$ with $i<j$ and $\mu_j>0$ and an exponent $1\leq s\leq \mu_{j}$
  (satisfying additionally $s\prec_{p}\mu_{j}$ if $\ch{\kk}=p>0$) such
  that
  \begin{equation} \label{eq:ssobstr}
    x_i^{\mu_j-s}\dfrac{\lt{g}}{x_j^{\mu_j-s}} \notin \lt{\I}\,.
  \end{equation}
  If $\ch{\kk}>0$, assume in addition that $\kk$ contains more than
  $\deg{g}$ elements.  Finally, let $\Psi:(x_j \mapsto x_j+ax_i)$ be a
  linear coordinate transformation with a parameter $a\in\kk^{\times}$ and
  $\tilde{G}$ the reduced Gr\"obner basis of the transformed ideal
  $\Psi(\I)$.  Then there exists at least one value $a\in\kk^{\times}$ such
  that
  \begin{displaymath}
    \ls(G) \prec_\ls \ls(\tilde{G})\,.
  \end{displaymath}
  In the case of an infinite coefficient field $\kk$, this estimate holds
  for a (Zariski) generic choice of $a$.
\end{proposition}

\begin{proof}
  Lemmata \ref{lem:step2} and \ref{lem:step1}, respectively, assert that
  \begin{displaymath}
    \ls(G) \preceq_\ls \ls\bigl(\Psi(G)^\Delta\bigr)\preceq_\ls 
    \ls(\tilde{G})\,.
  \end{displaymath}
  To prove our assertion, we show that \eqref{eq:ssobstr} implies that
  $\ls(G) \neq \ls\bigl(\Psi(G)^\Delta\bigr)$ for a suitable choice of the
  parameter $a$.  Let us assume that this was not the case.  Further let
  $G=\{g_1,\ldots,g_\ell\}$ and
  $\Psi(G)^\Delta=\{\hat{g}_1,\ldots,\hat{g}_\ell\}$.  Without loss of
  generality, suppose that $\lt{g_k} \prec_{\rl} \lt{g_l}$ and
  $\lt{\hat{g}_k} \prec_{\rl} \lt{\hat{g}_l}$ if $k>l$.  Our assumption
  implies that $\lt{g_k}=\lt{\hat{g}_k}$ for all $k$.  Suppose that $g=g_r$
  with $\lt{g_r}=\xv^\mu$ and denote
  $t=x_i^{\mu_j-s}\lt{g_r}/x_j^{\mu_j-s}$.  For $s=\mu_j$ the term $t$ was
  equal to $\lt{g_r}\in\lt{\I}$ contradicting \eqref{eq:ssobstr}.  Thus we
  may assume $s<\mu_j$ and then for the reverse lexicographic order
  $\lt{g_r}\prec_{\rl} t$.

  Lemma \ref{lem:obstrafo} asserts that for a suitable choice of $a$ every
  term of the form $x_i^{\mu_j-\hat{s}}\lt{g_r}/x_j^{\mu_j-\hat{s}}$ with
  $0\leq\hat{s}\leq\mu_j$ lies in the support of $\Psi(g_r)$, in particular
  $t\in\supp\bigl(\Psi(g_r)\bigr)$.  Since $\lt{g_r}=\lt{\hat{g}_r}$, any
  term in $\Psi(g_r)$ that is greater than $\lt{g_r}$ must be reduced.
  Since $t$ is one of these terms, there must be an element in
  $\{\lt{g_1},\ldots,\lt{g_\ell}\}$ that divides $t$.  But this means that
  $t\in\langle\lt{g_1},\ldots,\lt{g_\ell}\rangle=\lt{\I}$ which is a
  contradiction to \eqref{eq:ssobstr}.\qed
\end{proof}

\begin{remark}\label{rem:obstr_sqs}
  Proposition \ref{prop:list} encapsulates the central part of our
  termination proof.  As mentioned above, it is formulated for the case of
  strongly stable position.  Indeed, \eqref{eq:ssobstr} simply represents
  an obstruction to strong stability of the leading ideal $\lt{\I}$ (for
  $\ch{\kk}=p>0$ to strong $p$-stability).  With suitable adaptions, one
  easily obtains analogous propositions for any of the stable positions
  introduced in Section \ref{sec:combgener}.
\end{remark}

\begin{theorem}\label{thm:termination}
  If $\ch{\kk}=0$, then Algorithm \ref{algo:sstrafo} terminates after
  finitely many steps and returns a coordinate transformation $\Psi$ such
  that $\Psi(\I)$ is in strongly stable position.
\end{theorem}

\begin{proof}
  Let $\I$ be the given ideal and $G$ its reduced Gr\"obner basis.
  According to Remark \ref{rem:finlt}, $\I$ has only finitely many
  different leading ideals under linear coordinate transformations.  We
  denote the minimal bases of these leading ideals by
  $B_{1},\ldots,B_{\ell}$ and assume without loss of generality that
  \begin{displaymath}
    \ls(B_1)\prec_\ls \cdots \prec_\ls \ls(B_\ell)\,.
  \end{displaymath}
  In particular, there must be an index $1\leq\alpha\leq\ell$ such that
  $\lt{\I}=\langle B_{\alpha}\rangle$ and thus $\ls(G)=\ls(B_\alpha)$.

  If $\lt{\I}$ is not strongly stable, there exists a generator $g\in G$
  and integers $i,j \in \{1,\ldots,n\}$ with $i<j$ such that $x_j$ divides
  $\lt g=\xv^\mu$ and $x_i\lt{(g)}/x_j\notin\lt{\I}$.  Consider the
  transformation $\Psi_1:(x_j\mapsto x_j+x_i)$ and let $\tilde{G}_1$ be the
  reduced Gr\"obner basis of the transformed ideal $\Psi_1(\I)$.  There is
  an index $1\leq\beta\leq\ell$ such that
  $\lt{\Psi_1(\I)}=\langle B_{\beta}\rangle$ and thus
  $\ls(\tilde{G}_1)=\ls(B_\beta)$.  If $a=1$ is a generic value in
  Proposition \ref{prop:list}, then $\alpha<\beta$.  Otherwise, we enter
  the \texttt{while} loop in line \ref{line:wlbc} and perform the
  transformation $\Psi_1$ a second time.  The two transformations together
  are equivalent to the single transformation $(x_j\mapsto x_j+2x_i)$.
  Thus the effect of the inner \texttt{while} loop is that we try for the
  parameter $a$ consecutively the values $1,2,3,\ldots$ We know from
  Proposition~\ref{prop:list} that there are only a finite number of
  ``bad'' values of $a$ and thus after finitely many iterations we will
  reach a ``good'' one.  Hence there is an integer $r$ such that the
  reduced Gr\"obner basis $\tilde{G}_r$ of $\Psi_1^{r}(I)$ satisfies
  $\ls(\tilde{G}_r)=\ls(B_\gamma)$ with $\alpha<\gamma\leq\ell$.

  Since there are only finitely many different leading ideals possible, it
  is obvious that also the outer \texttt{while} loop is iterated only a
  finite number of times.  However, the termination of this loop is
  equivalent to the fact that the final transformed ideal is in a strongly
  stable position.\qed
\end{proof}

\begin{example}\label{ex:lose_ss}
  In the situation of the proof of Theorem \ref{thm:termination} one could
  be tempted to think that if $B_{\delta}$ is the minimal basis of a
  strongly stable leading ideal, then all bases $B_{\epsilon}$ with
  $\epsilon>\delta$ also generate strongly stable ideals.  This is,
  however, not true.  Consider the ideal
  \begin{displaymath}
    \I=\langle x_1^3,x_1^2x_2+x_1x_2^2+x_1x_3^2,x_1^2x_3,x_1^2x_4
    \rangle\unlhd\kk[x_1,x_2,x_3,x_4]\,.
  \end{displaymath}
  Its leading ideal
  \begin{displaymath}
    \lt{\I}= \langle x_1^3, x_1^2x_2, x_1^2x_3, x_1^2x_4, x_1x_2^3,
                     x_1x_2^2x_3, x_1x_2^2x_4\rangle 
  \end{displaymath}
  is strongly stable.  After the transformation
  $\Psi:(x_3 \mapsto x_3+x_2)$, we find
  \begin{displaymath}
    \lt{\Psi(\I)}= \langle x_1^3, x_1^2x_2, x_1x_2^2, x_1^2x_4,
                           x_1^2x_3^2\rangle 
  \end{displaymath}
  which is no longer strongly stable, as
  $x_3(x_1^2x_4)/x_4=x_1^2x_3\notin\lt{\Psi(\I)}$.  However,
  $\lt{\I}\prec_\ls\lt{\Psi(\I)}$.
\end{example}

\subsection{An Algorithm for Positive Characteristic}
\label{sec:algposchar}

The adaption of Algorithm \ref{algo:sstrafo} to a field $\kk$ of positive
characteristic $p$ faces two problems.  Firstly, the strategy for the
choice of the parameter $a$ realised by the inner \texttt{while} loop is no
longer valid, as it obviously fails as soon as the loop is iterated the
$p$th time.  If $\kk$ is an infinite field, then one uses simply an
enumeration of a countable subset of $\kk$, i.\,e.\ a procedure that
returns for each natural number $\ell\in\NN$ a different element
$a_{\ell}\in\kk$.  In the case of a finite field, one uses an enumeration
of the whole field.  Then the transformation $\Psi_{1}$ in the proof of
Theorem \ref{thm:termination} is defined as
$(x_{j}\mapsto x_{j}+a_{1}x_{i})$ and in the loop we do not apply the same
transformation again and again, but instead of the transformation
$\Psi_{1}^{\ell}$ we use $(x_{j}\mapsto x_{j}+a_{\ell}x_{i})$ in the
$\ell$th iteration.

Secondly, in positive characteristic all our auxiliary statements require
that the base field is sufficiently large (this also affects the modified
strategy for the inner \texttt{while} loop where one needs for each
iteration a new field element).  In each of the statements, it was
straightforward to specify precisely what the minimal required size is and
this number could be easily read off from the input data.  In the context
of Algorithm \ref{algo:sstrafo}, one can still easily state a bound: the
maximal degree of a generator in one of the minimal bases $B_{i}$.
However, since we do not compute the Gr\"obner system, we do not know this
number.  On the other hand, the bounds in the various lemmata and
propositions are worst case estimates and will in practice almost never be
realised.  Hence in an implementation one simply checks in the inner
\texttt{while} loop whether one still has new field elements to try.  If
this is not the case, one must perform a field extension.

These two modifications lead to Algorithm \ref{algo:bftrafo} for
transforming an ideal over a base field $\kk$ of characteristic $p$ into
strongly $p$-stable position (i.\,e.\ into $p$-Borel-fixed position).  It
uses an enumeration procedure \texttt{enum} for generating new field
elements as discussed above.  The proof of its correctness and termination
for a sufficiently large field is now completely analogous to the one of
Theorem \ref{thm:termination} and therefore omitted.  Again it is
straightforward to adapt the algorithm to other notions of $p$-stability.

\begin{algorithm}[ht]
  \caption{\textsc{BF-Trafo}: Transformation to $p$-Borel-fixed position
    with $\ch{\kk}=p>0$}\label{algo:bftrafo} 
  \begin{algorithmic}[1]
    \REQUIRE Reduced Gr\"obner basis $G$ of ideal $\I\unlhd\P$
    \ENSURE  a linear change of coordinates $\Psi$ such that
             $\lt{\Psi(\I)}$ is $p$-Borel-fixed
    \STATE $\Psi:=\mathrm{id}$
    \WHILE {$\exists\, g\in G,\ 1\leq j\leq n,\ 1\leq i < j,\ 1\leq s \leq 
          \mu_j\,: \newline 
          \hspace*{2cm} \displaystyle x_j\mid \lt{g}=\xv^\mu \wedge
          \binom{\mu_j}{u}\not\equiv 0 \mod p \wedge  	
          x_{i}\frac{\lt{g}}{x_{j}}\notin\langle\lt{G}\rangle$}
      \STATE $k:=1$;\quad 
             $\psi:=\bigl(x_{j}\mapsto x_{j}+\mathtt{enum}(k)x_{i}\bigr)$
      \STATE $\tilde{G}:=
             \text{\textsc{ReducedGr\"obnerBasis}}\bigl(\psi(G)\bigr)$;
      \WHILE {$\ls(G)\succeq_\ls \ls(\tilde{G})$}
        \STATE $k:=k+1$
        \IF{$k>|\kk|$}
          \STATE \textbf{error:} field too small
        \ELSE
          \STATE $\psi:=\bigl(x_{j}\mapsto x_{j}+\mathtt{enum}(k)x_{i}\bigr)$
          \STATE $\tilde{G}:=
                 \text{\textsc{ReducedGr\"obnerBasis}}\bigl(\psi(G)\bigr)$;
        \ENDIF
      \ENDWHILE
      \STATE $\Psi:=\psi\circ\Psi$; $G:=\tilde{G}$
    \ENDWHILE
    \RETURN $\Psi$
  \end{algorithmic}
\end{algorithm}

\begin{remark}
  We mentioned without justification in Section \ref{sec:combgener} that
  one does not need a
  $p$-version of quasi-stability.  The reason is simply that even in
  positive characteristic one can always reach a quasi-stable position.
  Indeed, if one considers the behaviour of a single term under the simple
  transformations we use, i.\,e.\ \eqref{eq:termtrafo}, then the idea
  underlying our algorithms is to replace the old term
  $\xv^{\nu}$ by one of the new terms appearing in
  $\Psi(\xv^{\nu})$.  For obtaining a (strongly) stable position, the
  relevant term is generally one ``in the middle'' of
  $\Psi(\xv^{\nu})$ and thus is multiplied by a binomial coefficient which
  may be zero in positive characteristic (the
  ``$p$-versions'' are defined in exactly such a way that these terms never
  become relevant).  For obtaining a quasi-stable position, we always need
  the last term whose binomial coefficient is one.  Thus even in positive
  characteristic we never encounter a problem, provided the field
  $\kk$ is sufficently large.
\end{remark}

\subsection{Implementations and Experiments}
\label{sec:impexp}

An efficient implementation of the algorithm described in this work is
highly non-trivial, as many aspects have to be considered.  A first point
concerns the strategy by which the next transformation is chosen, as often
several obstructions exist simulataneously and each may propose a different
elementary move.  Then one must decide whether one performs in each
iteration only one elementary move or whether one combines several moves
into a larger transformation.  Obviously, the first approach gives a better
chance to preserve sparsity while the second approach might reduce the
number of Gr\"obner bases computations.  These two points will require
extensive experiments.  We have mentioned already above that to many
natural strategies one can construct counter examples where it fares badly.
Hence only by experiments one can study the average behaviour for classical
examples typical for applications.

Finally, one must discuss how these repeated Gr\"obner bases computations
can be done most efficiently.  One should note that one always considers
the same ideal, however, in different coordinates.  Thus the question
arises how a Gr\"obner (or involutive) basis of an ideal in one coordinate
system can be efficiently transformed into one for the same ideal expressed
in another coordinate system.  In particular from an involutive basis, many
invariants of the ideal like its Hilbert function can be easily read off
and, in principle, one even knows a basis of the first syzygy module
\citep{wms:comb2}.  Thus ideas like a Hilbert-driven Buchberger algorithm
\citep{ct:hdba} or exploiting syzygies for the detection of reductions to
zero \citep{mmt:gbsyz} (see more generally \citep{ef:sign} for a recent
survey on signature based algorithms) can significantly increase the
efficiency.  \cite{bhs:compinv} report on some preliminary results in
particular concerning the first point.

As the design of a new specialised algorithm for computing Gr\"obner or
involutive bases is outside of the scope of this work, we only briefly
describe the results of four small test computations performed with a
prototype implementation\footnote{The code and the used examples are
  available at \url{http://amirhashemi.iut.ac.ir/software.html} (we
  therefore refrain from giving explicitly the generators).  To be
  consistent with the assumptions of this article, we homogenised all
  examples.} of our algorithm in \textsc{Maple}.  In this simple
implementation at each iteration the first found elementary move is taken
(with the leading terms sorted according to our term order).  Instead of
the strategy described here, a random integer value between $-2$ and $2$ is
chosen for the parameter $a$ (in our experience this suffices for small
examples as considered here). 

The following examples are taking from standard test suites for Gr\"obner
bases computations.  They can e.\,g.\ be found at
\url{http://invo.jinr.ru/ginv/}.  To demonstrate the flexibility of the
algorithm, we go in each example for a different generic position.

\begin{example}
  The Butcher ideal is generated by seven polynomials in eight variables
  with degrees up to $4$ and of dimension $3$.  Our implementation finds
  that the single elementary move $x_8\mapsto x_8-x_4$ transform it into
  Noether position.  By comparison, \textsc{Magma}'s command
  \texttt{NoetherNormalisation} delivers the much denser linear change of
  coordinates
  $x_6\mapsto x_6-2x_1 - x_2 - x_3,x_7\mapsto x_7+ 3x_2 + x_3 +
  x_5,x_8\mapsto x_8-3x_1 + 4x_2 - 2x_4 + 2x_5 + x_6 + x_7$ using the
  probabilistic method of \citep{gp:singular}.
\end{example}

\begin{example}
  The Vermeer ideal is generated by four polynomials in six variables with
  degrees up to $5$ and of dimension $3$.  Our implementation finds a
  single elementary move $x_6\longmapsto x_6+x_3$ to transform it into
  quasi-stable position where one could immediately read off many of its
  invariants from a Pommaret basis.
\end{example}

\begin{example}
  The Noon ideal is generated by four polynomials in five variables of
  degree $3$ and of dimension $1$.  For putting it into stable position,
  our implementation produces the following sequence of seven elementary
  moves: $x_4 \mapsto x_4+x_1$, $x_4 \mapsto x_4+2x_3$,
  $x_3 \mapsto x_3+x_1$, $x_3 \mapsto x_3+2x_2$, $x_4 \mapsto x_4+2x_3$,
  $x_5 \mapsto x_5-2x_1$, $x_5 \mapsto x_5-x_4$. In total this corresponds
  to a linear change with the matrix
  \begin{displaymath}
    A=
    \begin{pmatrix}
      1 & 0 & 0 & 0 & 0\\
      0 & 1 & 0 & 0 & 0\\
      1 & 2 & 1 & 0 & 0\\
      2 & 4 & 4 & 1 & 0\\
      -1 & 4 & 4 & -1 & 1\\
    \end{pmatrix}\,.
  \end{displaymath}
  Thus here we obtain an almost dense lower triangular matrix which more or
  less represents the worst case for our algorithm.  At least the
  coefficients are very small.  By contrast, a call of \textsc{CoCoA}'s
  command \texttt{gin} yields usually a linear transformation which
  consists of a dense lower triangular matrix where each non-zero entry is
  an integer with five to six digits.  This is a typical behaviour for
  probabilistic approaches.
\end{example}

\begin{example}
  The Weispfenning94 ideal is generated by three polynomials in four
  variables with degrees up to $5$ and of dimension $2$. For putting it
  into strongly stable position, our implementation produces the following
  sequence of four elementary moves: $x_2 \mapsto x_2-x_1$,
  $x_4 \mapsto x_4-2x_3$, $x_3 \mapsto x_3+x_1$, $x_4 \mapsto x_4+2x_3$.
  In total this corresponds to a linear change with the matrix
  \begin{displaymath}
    A=
    \begin{pmatrix}
      1 & 0 & 0 & 0\\
      -1 & 1 & 0 & 0\\
      1 & 0 & 1 & 0\\
      2 & 0 & 0 & 1\\
    \end{pmatrix}\,.
  \end{displaymath}
  Thus this time we end up with a fairly sparse transformation.  A
  probabilistic computation indicates that it actually even yields
  $\gin{\I}$.
\end{example}

\section*{Acknowledgements}  

The first author would like to thank DAAD (German Academic Exchange
Service) for supporting his stays at Universit\"at Kassel in 2013 and 2016
during which much of the work on this article was done.  He would also like
to thank Professor W.M. Seiler for the invitation, hospitality, and
support.  The research of the first author was in part supported by a grant
from IPM (No.~92550420).  The work of the third author was partially
performed as part of the H2020-FETOPEN-2016-2017-CSA project $SC^{2}$
(712689).  Finally, the authors would like to thank the referees for their
very detailed and constructive comments.

\bibliographystyle{plainnat}

\end{document}